\documentclass[sigconf, nonacm]{acmart}
\usepackage{xcolor}

\newcommand\vldbdoi{XX.XX/XXX.XX}
\newcommand\vldbpages{XXX-XXX}
\newcommand\vldbvolume{19}
\newcommand\vldbissue{7}
\newcommand\vldbyear{2026}
\newcommand\vldbauthors{\authors}
\newcommand\vldbtitle{\shorttitle} 
\newcommand\vldbavailabilityurl{https://github.com/pkumod/CEMR}
\newcommand\vldbpagestyle{empty} 

\newbool{fullversion}
\booltrue{fullversion}   

\usepackage{amsmath,amsfonts}
\usepackage{textcomp}
\usepackage{algorithmic}
\usepackage{graphicx}
\usepackage[noend, ruled,linesnumbered]{algorithm2e}
\usepackage{subcaption}
\usepackage{threeparttable}
\usepackage{soul}
\usepackage{multicol}
\usepackage{multirow}
\usepackage{bbm}
\usepackage{amsthm}
\usepackage{pifont}
\usepackage{enumitem}
\usepackage{tcolorbox}
\usepackage{listings}
\tcbuselibrary{listings,skins}
\usepackage{etoolbox}
\usepackage{dblfloatfix}

\newcommand{\nop}[1]{}
\newtheorem{example}{Example}
\newtheorem{lemma}{Lemma}

\usepackage{xspace}
\def\method{\texttt{CEMR}\xspace}
\def\plainmethod{CEMR\xspace}

\newcommand{\secref}[1]{Section \ref{#1}}
\newcommand{\figref}[1]{Figure \ref{#1}}
\newcommand{\tabref}[1]{Table \ref{#1}}

\newcommand{\algref}[1]{Algorithm \ref{#1}}
\newcommand{\thmref}[1]{Theorem \ref{#1}} 
\newcommand{\lemref}[1]{Lemma \ref{#1}}

\newcommand{\exgref}[1]{Example \ref{#1}}
\newcommand{\appref}[1]{Appendix \ref{#1}}

\newcounter{redcount} 
\newcommand{\red}[1]{%
  \stepcounter{redcount}
  \begingroup\color{red}#1\endgroup
}

\newcounter{bluecount} 

\SetCommentSty{mycommfont}

\newcommand{\nosection}[1]{\noindent\textbf{\emph{#1}}}

\newlength{\origtextfloatsep}

\begin{document}
\ifbool{fullversion}{
\title{CEMR: An Effective Subgraph Matching Algorithm with Redundant Extension Elimination (Full Version)}
}{
\title{CEMR: An Effective Subgraph Matching Algorithm with Redundant Extension Elimination}
}



\author{Linglin Yang}
\authornote{Linglin Yang and Xunbin Su contributed equally to this work.}
\affiliation{%
  \institution{Peking University}
  \city{Beijing}
  \country{China}
}
\email{linglinyang@stu.pku.edu.cn}

\author{Xunbin Su}
\authornotemark[1]
\affiliation{%
  \institution{Peking University}
  \city{Beijing}
  \country{China}
}
\email{suxunbin@pku.edu.cn}

\author{Lei Zou}
\affiliation{%
  \institution{Peking University}
  \city{Beijing}
  \country{China}
}
\email{zoulei@pku.edu.cn}

\author{Xiangyang Gou}
\affiliation{%
  \institution{University of New South Wales}
  \city{Sydney}
  \country{Australia}
}
\email{xiangyang.gou@unsw.edu.au}

\author{Yinnian Lin}
\affiliation{%
  \institution{Peking University}
  \city{Beijing}
  \country{China}
}
\email{linyinnian@pku.edu.cn}

\begin{abstract}

Subgraph matching is a fundamental problem in graph analysis with a wide range of applications. However, due to its inherent NP-hardness, enumerating subgraph matches efficiently on large real-world graphs remains highly challenging. Most existing works adopt a depth-first search (DFS) backtracking strategy, where a partial embedding is gradually extended in a DFS manner along a branch of the search trees until either a full embedding is found or no further extension is possible. A major limitation of this paradigm is the significant amount of duplicate computation that occurs during enumeration, which increases the overall runtime. To overcome this limitation, we propose a novel subgraph matching algorithm, \method. It incorporates two techniques to reduce duplicate extensions: \textit{common extension merging}, which leverages a black-white vertex encoding, and \textit{common extension reusing}, which employs common extension buffers. In addition, we design two pruning techniques to discard unpromising search branches. Extensive experiments on real-world datasets and diverse query workloads demonstrate that \method outperforms state-of-the-art subgraph matching methods.

\end{abstract}

\maketitle

\pagestyle{\vldbpagestyle}
\begingroup\small\noindent\raggedright\textbf{PVLDB Reference Format:}\\
\vldbauthors. \vldbtitle. PVLDB, \vldbvolume(\vldbissue): \vldbpages, \vldbyear.\\
\href{https://doi.org/\vldbdoi}{doi:\vldbdoi}
\endgroup
\begingroup
\renewcommand\thefootnote{}\footnote{\noindent
This work is licensed under the Creative Commons BY-NC-ND 4.0 International License. Visit \url{https://creativecommons.org/licenses/by-nc-nd/4.0/} to view a copy of this license. For any use beyond those covered by this license, obtain permission by emailing \href{mailto:info@vldb.org}{info@vldb.org}. Copyright is held by the owner/author(s). Publication rights licensed to the VLDB Endowment. \\
\raggedright Proceedings of the VLDB Endowment, Vol. \vldbvolume, No. \vldbissue\ %
ISSN 2150-8097. \\
\href{https://doi.org/\vldbdoi}{doi:\vldbdoi} \\
}\addtocounter{footnote}{-1}\endgroup

\ifdefempty{\vldbavailabilityurl}{}{
\vspace{.3cm}
\begingroup\small\noindent\raggedright\textbf{PVLDB Artifact Availability:}\\
The source code, data, and/or other artifacts have been made available at \url{https://github.com/pkumod/CEMR}.
\endgroup
}

\section{Introduction}
\label{sec:introduction}

In recent years, graph structures have gained increasing significance in the data management community. As one of the fundamental tasks in graph analysis, subgraph matching has attracted much attention from both industry and academia \cite{sahu2017ubiquity, livi2013graph, mahmood2017large} and is widely used in many applications, such as chemical compound search \cite{yan2004graph}, social network analysis \cite{fan2012graph}, protein-protein interaction network analysis \cite{kim2011biological} and RDF query processing \cite{zou2011gstore, kim2015taming}. 

Given a data graph $G$ and a query graph $Q$, subgraph matching (isomorphism) aims to find all the subgraphs of $G$ that are isomorphic to $Q$. The mapping between each of these subgraphs and $Q$ is called an embedding (or match) of $Q$ over $G$. For example, given a data graph $G$ and a query graph $Q$ in \figref{fig:graph_example}, \{($u_0$, $v_0$), ($u_1$, $v_1$), ($u_2$, $v_2$), ($u_3$, $v_5$), ($u_4$, $v_{8}$), ($u_5$, $v_{10}$), ($u_6$, $v_{11}$)\} is one embedding of $Q$ over $G$. Unfortunately, subgraph matching is a well-known NP-hard problem \cite{garey1979guide} and the data graph $G$ is usually very large in practical applications. Thus, it is challenging to efficiently enumerate all embeddings of $Q$ over $G$.

\begin{figure}[htbp]
    \centering
    \vspace{-0.2cm}
    \setlength{\abovecaptionskip}{0.3cm}
    \setlength{\belowcaptionskip}{-0.1cm}
    \begin{subfigure}[c]{0.24\linewidth}
        \includegraphics[width=\linewidth]{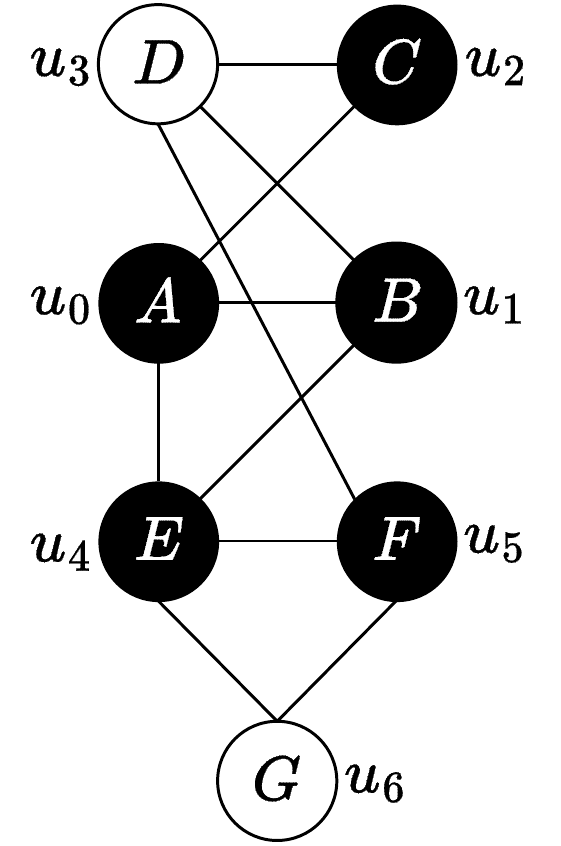}
        \captionsetup{skip=1.6pt}
        \caption{Query graph}
        \label{fig:graph_example_query}
    \end{subfigure}
    \hspace{0.2cm}
    \begin{subfigure}[c]{0.54\linewidth}
        \includegraphics[width=\linewidth]{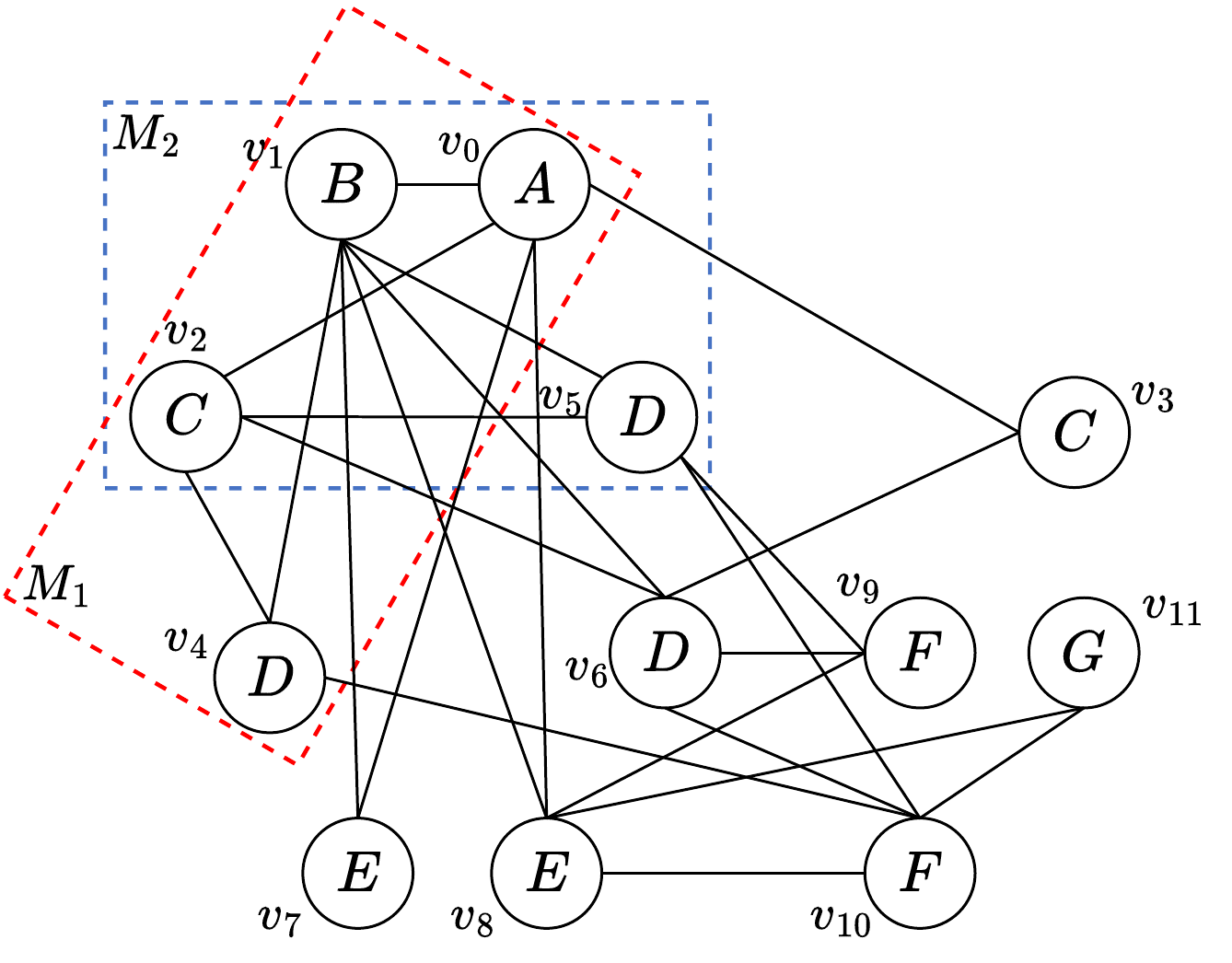}
        \captionsetup{skip=1.6pt}
        \caption{Data graph}
        \label{fig:graph_example_data}
    \end{subfigure}
    
    \begin{subfigure}[b]{\linewidth}
        \centering
        \includegraphics[width=0.83\linewidth]{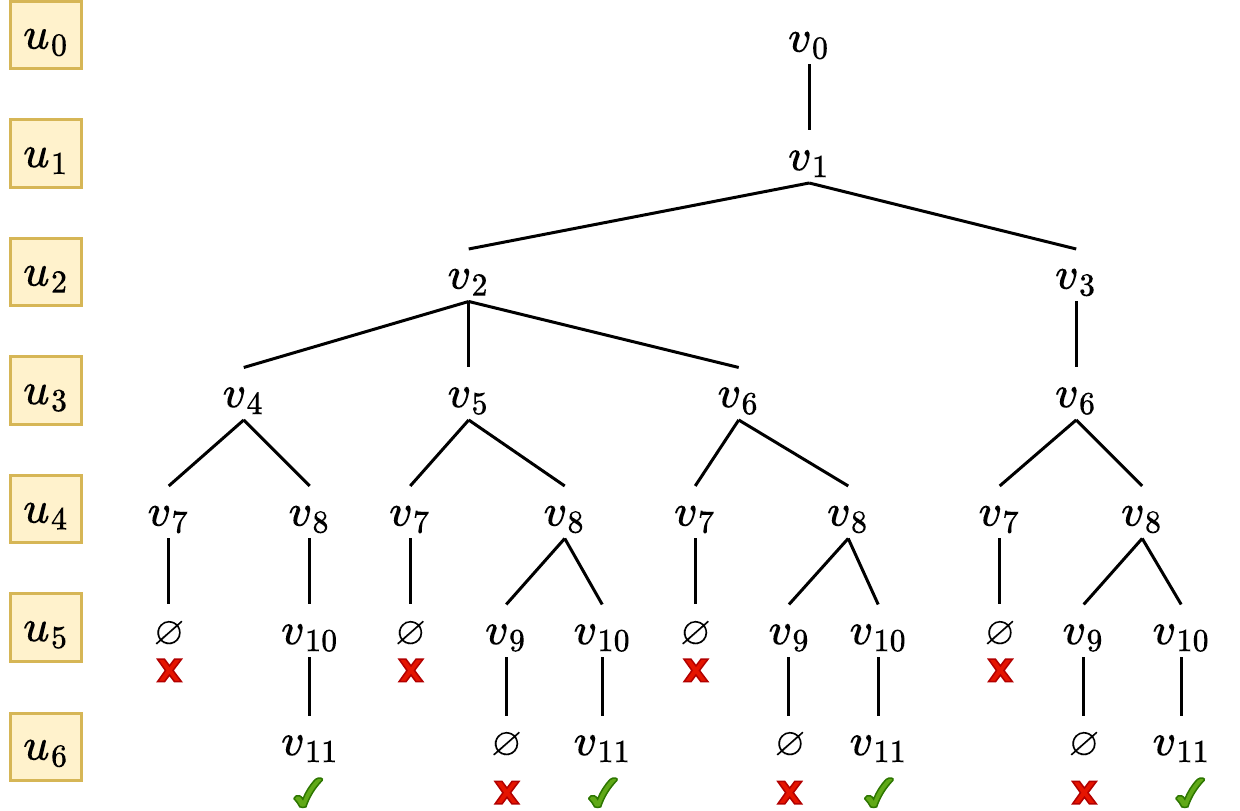}
        \captionsetup{skip=5pt}
        \caption{Search tree rooted at $\{(u_0, v_0)\}$}
        \label{fig:search_tree_example}
    \end{subfigure}
    \caption{Subgraph Matching}
    \label{fig:graph_example}
\end{figure}

Many algorithms \cite{shang2008taming, zhao2010graph, zhang2009gaddi, ren2015exploiting, cordella2004sub, han2013turboiso, bi2016efficient, kim2021versatile, han2019efficient, sun2020rapidmatch, jin2023circinus} have been proposed to solve the subgraph matching problem. A typical subgraph matching algorithm framework finds matches of the query graph $Q$ over the data graph $G$ by iteratively mapping each query vertex in $Q$ to data vertices in $G$ following a certain matching order. Such a matching process can be organized as a backtracking search on a search tree. 

For example, \figref{fig:search_tree_example} shows a search tree when mapping $u_0$ to $v_0$ from the previous example. Each tree node at level $i$ ($i\in\{0,\dots, |V(Q)|-1\}$) represents an embedding $M$ that matches the first $i+1$ query vertices in the matching order.
A green tick denotes a valid embedding, whereas a red cross denotes that the corresponding tree node cannot produce a valid embedding. 
We call an embedding at level $|V(Q)|-1$ \emph{full embedding} and an embedding at level less than $|V(Q)|-1$ \emph{partial embedding}. By adding matches of $u_{i+1}$ with different data vertices into a partial embedding $M$ at level $i$, we can extend $M$ to its children at level $i+1$. 

Most existing solutions adopt the \emph{depth-first search (DFS)} backtracking approach. They gradually extend one partial embedding following a certain branch in the search tree (i.e., perform DFS on the search tree) until a full embedding is found or no extension can be made. Then they backtrack to the upper level and try to search another branch. 
To speed up subgraph matching, existing solutions have proposed many heuristic optimization techniques on getting good matching orders \cite{shi2020graphpi, bi2016efficient, han2019efficient, juttner2018vf2++}, generating the smaller candidate space \cite{sun2020rapidmatch, han2013turboiso, han2019efficient, bhattarai2019ceci, bi2016efficient}, sharing computation among query vertices with equivalence relationship \cite{han2013turboiso, kim2021versatile, ren2015exploiting} and reducing duplicated computation \cite{ammar2018distributed, yang2021huge, trigonakis2021adfs, guo2020exploiting, jin2023circinus, mhedhbi2019optimizing, jamshidi2020peregrine, guo2020gpu}. 
Some of the remaining works \cite{shao2014parallel, lai2015scalable, afrati2013enumerating} adopt the \emph{breadth-first search (BFS)} approach where they find all partial embeddings at the same level of the search tree and then move on to the next level. \emph{BFS-based} solutions suffer from a large amount of memory consumption for storing intermediate results \cite{teixeira2015arabesque, chen2018g, jin2023circinus}. Furthermore, in cases where a query cannot be completed within reasonable time, \emph{DFS-based} methods are capable of obtaining a portion of the final embeddings within the given time constraint, whereas \emph{BFS-based} methods may yield no results at all. In this paper, we focus on optimizing \emph{DFS-based} approaches. 

\subsection{Motivation}

During the enumeration phase, there may exist duplicate extension computations. These duplicate computations usually occur at the same level of the search tree, where the backward neighbors of the next query vertex $u$ share the same mappings across different partial embeddings. The root cause is that the extension of $u$ depends only on its backward neighbors and is independent of other query vertices. These repeated calculations enlarge the search space and hinder enumeration performance. 
We use the following example to illustrate this phenomenon.

\begin{example}
\label{ex:motivation2}
Consider the data graph and query graph in \figref{fig:graph_example}, with matching order $O = (u_0, u_1, u_2, u_3, u_4, u_5, u_6)$.
Examine two partial embeddings:
$M_1 = \{(u_0, v_0), (u_1, v_1), (u_2, v_2), (u_3, v_4)\}$ (red dashed box) and
$M_2 = \{(u_0, v_0), (u_1, v_1), (u_2, v_2), (u_3, v_5)\}$ (blue dashed box).
Both share the prefix $(u_0, u_1) \mapsto (v_0, v_1)$. Notably, when extending to $u_4$, $M_1$ and $M_2$ perform the same extension and generate the same candidate set $\{v_7, v_8\}$, since $u_4$ only connects to $u_0$ and $u_1$. This redundancy implies that extending $M_2$ could reuse the results from $M_1$, allowing simultaneous or shared extension and avoiding repeated computation.
\end{example}

In BFS-based enumeration, it is straightforward to eliminate such redundancy by grouping partial embeddings that share the same extension patterns at the same level of the search tree, since a complete intermediate result table can be constructed in each iteration. However, the memory cost of BFS can be prohibitively high.
In contrast, when using DFS for enumeration, it becomes challenging to share common extension results among partial embeddings, as the repeated extensions may appear in different branches of the search tree. This necessitates a mechanism to reduce duplicate extension computations during the backtracking process.

\subsection{Our Solution}

To eliminate computation redundancy, we propose two core optimization techniques: \emph{common extension merging} (\texttt{CEM}) and \emph{common extension reusing} (\texttt{CER}), which together form our method \texttt{CEMR} (\textbf{C}ommon \textbf{E}xtension \textbf{M}erge and \textbf{R}eusing).

\texttt{CEM} enables joint extension of multiple search branches by merging them until divergence. To support this, we introduce a \emph{black-white vertex encoding} scheme (\secref{sec:merge}), which partitions query vertices into black and white. 
A white vertex can match multiple data vertices within a single partial embedding, e.g., treating $u_4$ as white resolves the issue illustrated in \exgref{ex:motivation2}.

\texttt{CER} leverages \emph{common extension buffers} (CEBs, see \secref{sec:cer_for_reusing}) to cache reusable extension results that can be shared across multiple partial embeddings, thereby reducing redundant computation.

In addition, we design two pruning techniques that efficiently identify and discard unpromising search branches during the backtracking enumeration process.

To summarize, we make the following contributions in this paper. 

\begin{itemize}[leftmargin=10pt]
    
    \item We propose a DFS-based subgraph matching algorithm, \method, that aims to reduce redundant computation during the enumeration phase.

    \item We develop a forward-looking common extension merging technique based on black-white vertex encoding, which effectively merges search branches with similar expansion behaviors. In addition, we propose a cost-driven encoding strategy designed to maximize computational cost reduction.

    \item We propose a backward-looking common extension reusing technique that caches and reuses extension results to avoid repeated computation.
    
    \item We introduce two pruning techniques that effectively eliminate unpromising search branches and redundant partial embeddings.
    
    \item We conduct extensive experiments on several real-world graph datasets. Results show that \method consistently outperforms state-of-the-art algorithms.
\end{itemize}

\section{Preliminary and Related Works}
\label{sec:preliminary}

\subsection{Problem Definition}
\label{sec:problem_definition}

In this paper, we focus on undirected, vertex-labeled, and connected graphs. Note that our techniques can be extended to more general cases (e.g., directed graphs and edge-labeled graphs, we provide some discussions on these extensions in \secref{sec:extensions}).

We use $G$ and $Q$ to denote the data graph and the query graph, respectively.
A (data) graph $G$ is a quadruple $G = (V, E, \Sigma, L_G)$, where $V(G)$ is the set of vertices, $E(G) \subseteq V(G) \times V(G)$ is the set of edges, $\Sigma$ is the set of labels, and $L_G$ is a labeling function that assigns a label from $\Sigma$ to each vertex in $V$.
Given $v \in V(G)$, $N(v)$ denotes the set of neighbors of $v$, i.e., $N(v) = \{v' \mid e(v, v') \in E(G)\}$. The degree of $v$, denoted by $d(v)$, is defined as $d(v) = |N(v)|$.
The query graph $Q$ is defined in the same way.

\begin{definition}[Subgraph Isomorphism]\label{def:subgraphisomorphsim}
Given a query graph $Q$ and a data graph $G$, $Q$ is subgraph isomorphic to $G$ if there exists an injective mapping function $M$ from $V(Q)$ to $V(G)$ such that:
\begin{itemize}[leftmargin = 10pt]
    \item (Label constraint) $\forall u \in V(Q), L_Q(u)=L_G(M[u])$;
    \item (Topology constraint) $\forall e(u, u')\in E(Q), e(M[u], M[u']) \in E(G)$.
\end{itemize}
Sometimes, the mapping function $M$ is also called an embedding.
\end{definition}

\begin{definition}[Subgraph Matching Problem]
Given a query graph $Q$ and a data graph $G$, the subgraph matching problem is to find all distinct subgraphs (\emph{embeddings}) of $G$ that are \emph{isomorphic} to $Q$.
\end{definition}

An embedding of an induced subgraph of $Q$ in $G$ is called a \textit{partial embedding}. In contrast, an embedding of the complete query graph $Q$ is referred to as a \textit{full embedding}.

\begin{definition}[Matching Order]
A matching order is a permutation of the query vertices, denoted as $O = (u_0, u_1, \dots, u_{n-1})$, where we assume $|V(Q)| = n$ throughout the rest of the paper. A valid matching order must preserve the connectivity of the query graph. That is, for each $i \in \{1, \dots, n - 1\}$, the subquery $Q_i$ induced by the first $i$ vertices $\{u_0, u_1, \dots, u_{i-1}\}$ in $O$ must be connected. For clarity, we consistently use this notation throughout the paper.
\end{definition}

\begin{definition}[Backward (Forward) Neighbors]
\label{def:BF_neighbors}
Given a matching order $O$, the set of backward neighbors $N^O_{-}(u)$ (forward neighbors $N^O_{+}(u)$) of a query vertex $u$ is the set of neighbors of $u$ whose indices in $O$ are smaller (greater) than $u$'s corresponding index $o(u)$ in $O$. Formally, $N^O_{-}(u_i) = \{u\mid u_j \in N(u), j < i\}$ and $N^O_{+}(u_i) = \{u_j\mid u_j \in N(u), j > i\}$.
\end{definition}

\begin{example} Considering the query graph $Q$ in \figref{fig:graph_example_query}, assume that the matching order is $O = (u_0, u_1, u_2, u_3, u_4, u_5, u_6)$. For $u_4$, it has four neighbors $u_0$, $u_1$, $u_5$, and $u_6$. Among them, $u_0$ and $u_1$ are matched before $u_4$, which are backward neighbors of $u_4$ (i.e., $N^O_{-}(u_4)=\{u_0, u_1\}$). In contrast, $u_5$ and $u_6$ are forward neighbors of $u_4$ (i.e., $N_{+}^O(u_4)=\{u_5, u_6\}$).
\end{example}

\tabref{tab:notations} lists the frequently used notations in this paper.

\begin{small}
\begin{table}[htbp]
    \centering
	\caption{Frequently used notations.}
	\label{tab:notations}
	\begin{tabular}{c|c}
		\toprule
		Notation & Description\\
		\midrule
		$G$, $Q$ & Data graph and query graph\\
		$v$, $u$ & Data vertex and query vertex\\
        $O$ & Matching order \\
        $Q_i$ & Induced subquery by the first $i$ vertices \\
        $N^O_{-}(u)$, $N^O_{+}(u)$ & Backward and forward neighbors of $u$ given $O$\\
        $BK(u)$, $WT(u)$ & Black and white backward neighbors of $u$\\
            $M$, $M[u]$ & A (partial) embedding , mapping of $u$ in $M$\\
		$C(u)$ & Candidate vertex set of $u$\\
            $\mathcal{A}$ & Auxiliary data structure \\
            $\mathcal{A}_{u'}^u(v)$ & Neighbors of $v$ in $C(u')$ where $v\in C(u)$\\
		$R_M(u)$ & Extensible vertices of $u$ under $M$\\
            $RS(u)$ & Reference set of $u$\\
            $Con(u)$ & Contained vertex set of $u$\\
		\bottomrule
	\end{tabular}
\end{table}
\end{small}

\subsection{Background and Related Work}\label{sec:relatedwork}

Most recent studies \cite{han2013turboiso, bi2016efficient, bhattarai2019ceci, han2019efficient, sun2020rapidmatch, kim2021versatile, arai2023gup, choi2023bice, lu2025b} on subgraph matching are based on the \textit{preprocessing-enumeration} framework \cite{sun2020memory} as shown in \algref{algo:generic_framework}. They first find candidates of query vertices and edges to build an online auxiliary structure, then enumerate all embeddings with the help of this auxiliary structure.

\begin{small}
\begin{algorithm}[t]
	\caption{Generic framework of \emph{indexing-enumeration}}
	\label{algo:generic_framework}
	\SetAlgoLined
    \DontPrintSemicolon
    \SetKwFunction{Enumerate}{Enumerate}
	\KwIn{Query graph $Q$, data graph $G$.}
	\KwOut{All embeddings of $Q$ in $G$, denoted as $\mathcal{M}$.}
	$C$, $\mathcal{A}$ $\leftarrow$ generate candidate vertices and an auxiliary index\; 
	$O$  $\leftarrow$ get the matching order\; 
    $\mathcal{M} \leftarrow \emptyset$\; 
	Enumerate($Q$, $\mathcal{A}$, $O$, $\mathcal{M}$, $\{\}$, $0$)\;
	\textbf{return} $\mathcal{M}$\;
\end{algorithm}
\end{small}

\subsubsection{Preprocessing Phase}

In Algorithm \ref{algo:generic_framework}, the preprocessing phase consists of filtering and matching order generation.
Specifically, it first generates a candidate vertex set $C$ ($C(u)$ for each $u \in V(Q)$) and builds an auxiliary data structure $\mathcal{A}$ to maintain the candidate edges between the candidate vertex sets (line 1). 
Then, the algorithm generates a matching order $O$ usually in a greedy way based on some heuristic rules (line 2).

\subsubsection{Enumeration Phase}

After the preprocessing phase, \algref{algo:generic_framework} recursively enumerates matches along the matching order $O$ (line 4). The enumeration procedure is detailed in \algref{algo:set_intersection_framework}, which consists of two key steps. Given the current partial embedding $M$, line 3 computes all extensible vertices $R_M(u_i)$ for the next query vertex $u_i$. Then, lines 4-5 recursively extend the partial embedding by exploring all candidates in $R_M(u_i)$.

\begin{small}
\begin{algorithm}[htbp!]
	\caption{Enumerate($Q$, $\mathcal{A}$, $O$, $\mathcal{M}$, $M$, $i$)\ (Basic)}
	\label{algo:set_intersection_framework}
	\SetAlgoLined
    \DontPrintSemicolon
    \KwIn{The query $Q$, auxiliary structure $\mathcal{A}$, matching order $O$, result set $\mathcal{M}$, an embedding $M$ of $Q_{i}$, and the backtracking depth $i$.}
	\If{$i$ = $|V(Q)|$}{
	    $\mathcal{M} \leftarrow \mathcal{M}\cup \{M\}$, \textbf{return}\;
	}
	$R_M(u_i) = \cap_{u \in N^O_{-}(u_i)}$ $\mathcal{A}_{u_i}^u(M[u])$\;
	\ForEach{$v \in R_M(u_i)$}{
	    Enumerate($Q$, $\mathcal{A}$, $O$, $\mathcal{M}$, $M \cup \{(u_i, v)\}$, $i+1$)\;
        
	}
\end{algorithm}
\end{small}

\subsubsection{Related Works of Preprocessing Phase}

The label degree filter (LDF) and neighbor label filter (NLF) \cite{zhu2012treespan} are the two most widely used filtering methods. They check the label and degree of a data vertex and its neighborhood to eliminate vertices that cannot be matched to a query vertex.
In addition to vertex candidates, CFL \cite{bi2016efficient} also builds an auxiliary structure to maintain the edge candidates of a BFS tree in the query graph, enabling more powerful filtering. CECI \cite{bhattarai2019ceci} and DAF \cite{han2019efficient} further consider the non-tree edges of the query graph in their auxiliary structures.
RM \cite{sun2020rapidmatch} utilizes the semi-join operator to quickly eliminate dangling tuples. VC \cite{sun2020subgraph} and VEQ \cite{kim2021versatile} adopt more sophisticated filtering rules to produce tighter candidate sets, at the cost of increased filtering time.

Most existing subgraph matching methods use heuristic vertex ordering strategies.
RI \cite{bonnici2013subgraph} and RM \cite{sun2020rapidmatch} determine the query vertex order solely based on the structure of the query graph, prioritizing vertices from its dense regions.
Other methods adopt different heuristics, such as selecting vertices with higher degrees and smaller candidate sets \cite{bi2016efficient, he2008graphs}.
DAF \cite{han2019efficient} and VEQ \cite{kim2021versatile} further propose adaptive ordering strategies, where the choice of the next query vertex may vary across different partial embeddings.

\subsubsection{Related Works of Enumeration Phase}

As the enumeration phase is usually the most time-consuming part of subgraph matching, many works have proposed various optimizations to accelerate it, mainly by reducing redundant computation. These optimizations can be broadly classified into two categories. 
The first is backward-looking optimization \cite{han2019efficient, arai2023gup, kim2021versatile, sun2020rapidmatch}, which caches intermediate results to avoid redundant extensions.
The second is forward-looking optimization \cite{jin2023circinus, li2024newsp, yang2023fast}, which merges multiple search branches and extends them simultaneously in subsequent steps.

\textbf{Backward-Looking Optimization}. 
Backward-looking optimization leverages historical extensions to accelerate subsequent ones. It typically serves as a pruning technique to discard unpromising search branches based on past extension results.

DAF \cite{han2019efficient} proposes a \emph{failing set} pruning strategy to eliminate unpromising search branches. The basic idea is: Given a partial embedding $M$, if $M$ cannot be extended to form a complete match of $Q$, some other search branches (satisfying the failing set conditions) can be skipped without further computation. This technique is later adopted by RM \cite{sun2020rapidmatch} and VEQ \cite{kim2021versatile}.

GuP \cite{arai2023gup} extends the failing set idea by introducing guard-based pruning, which retains discovered unpromising partial matches for repeated pruning at the cost of additional memory usage. However, these methods focus solely on pruning invalid extensions and do not reuse valid historical extension results.

VEQ \cite{kim2021versatile} proposes a pruning strategy based on dynamic equivalence, which can eliminate equivalent search subtrees, including both promising and unpromising ones. For each query vertex $u$, a dynamic equivalence class is defined over the filtered candidate set $C(u)$. Two candidates $v_i$ and $v_j$ in $C(u)$ belong to the same class if they share common neighbors, i.e., $\mathcal{A}_{u'}^{u}(v_i)=\mathcal{A}_{u'}^{u}(v_j)$ for every $u' \in N(u)$. This idea was later also adopted by BICE \cite{choi2023bice}. However, its effectiveness depends on the number of equivalent candidate vertices. As a result, VEQ may perform poorly on data graphs with high average degrees or limited equivalence among candidates. 

\textbf{Forward-Looking Optimization}. 
Forward-looking optimization does not utilize historical results. Instead, it directly merges certain extension branches when it predicts that they will produce similar extension results.

Circinus \cite{jin2023circinus} proposes a merging strategy based on a vertex cover of the query graph.
For a query vertex $u_i$ not included in the vertex cover, all of its backward neighbors must be in the cover.
In such cases, given a partial embedding $M$, Circinus merges the candidate mappings of $u_i$ and compresses multiple extended partial embeddings into a group, delaying their expansion until $u_i$'s forward neighbors are processed.
On the other hand, if $u_i$ is in the vertex cover and has at least one backward neighbor not in the cover, then the compressed embeddings will be decomposed.
Specifically, the merged mapping set of this backward neighbor will be split according to the different mappings of $u_i$, and each decomposed embedding will be handled separately.
A continuous subgraph matching method, CaLiG \cite{yang2023fast}, also adopts this strategy.

BSX \cite{lu2025b} models the search space as a multi-dimensional search box. During the backtracking enumeration phase, it iteratively selects one dimension, which corresponds to a query vertex from the vertex cover, and chooses several similar data vertices for that dimension. Then, it refines the other dimensions accordingly. At the final enumeration layer, the homomorphic matches within the search box are verified and form the final matching results.

Some other works aim to identify unmatchable search branches in advance. For example, BICE \cite{choi2023bice} prunes branches that will result in conflicts by applying bipartite matching on the bipartite graph between unmapped query vertices and their candidate data vertices.

\section{Theoretical Foundation}
\label{sec:overview}

In this paper, we propose \underline{C}ommon \underline{E}xtension \underline{M}erging (\texttt{CEM}) and \underline{C}ommon \underline{E}xtension \underline{R}eusing (\texttt{CER}) techniques, which can be viewed as \emph{forward-looking} and \emph{backward-looking} optimizations, respectively.
Before diving into the details, we first study the following lemma, which lays the foundation for both \texttt{CEM} and \texttt{CER}.
\ifbool{fullversion}{The proof can be found in Appendix \ref{proof:redundant}.}{\red{The proof can be found in the full version of our paper \cite{yang2025neuso}.}}

\begin{lemma}
\label{lemma:redundant}
Given a query graph $Q$ and a matching order $O = (u_0, u_1, \dots, u_{n-1})$, let $M_1$ and $M_2$ be matches of the subquery $Q_i$ induced by the first $i$ vertices. Assume that $M_1$ and $M_2$ share the same mappings for all backward neighbors of $u_i$ (i.e., $N_-^O(u_i)$), and let $v$ be a data vertex that does not conflict with any vertex in $M_1$ or $M_2$. Then, if $M_1 \oplus v$ is a valid match of the subquery $Q_{i+1}$, so is $M_2 \oplus v$. Here, $M \oplus v$ denotes extending $M$ along the matching order by appending the mapping $(u_{|M|}, v)$.
\end{lemma}

Let us revisit \exgref{ex:motivation2}. Two partial embeddings $M_1$ and $M_2$ for $Q_4$ share the same mappings for $u_0$ and $u_1$, which are the backward neighbors of $u_4$. Therefore, for each candidate vertex $v$ of $u_4$ (i.e., $v \in \{v_7, v_8\}$), if $M_1 \oplus v$ forms a valid match of the subquery $Q_5$, then so does $M_2 \oplus v$.

\lemref{lemma:redundant} offers the most effective opportunity to eliminate redundant extensions during enumeration, as all valid extensions of $u_i$ for $M_1$ can be directly reused for $M_2$ without any additional computation. 
In BFS-based approaches, when extending $u_i$, we can group all partial embeddings that share the same mappings for $u_i$'s backward neighbors to avoid repeated work. However, doing so is non-trivial in DFS-based methods.
Our proposed techniques, \texttt{CEM} (\secref{sec:merge}) and \texttt{CER} (\secref{sec:reuse}), serve as practical relaxations of this lemma to facilitate extension sharing in DFS-based enumeration.

\section{Common Extension Merging}
\label{sec:merge}

The intuition behind \texttt{CEM} is to merge multiple partial embeddings into an aggregated embedding and extend them together. For instance, if $u_i$ is not connected to $u_{i+1}$, all mapped data vertices of $u_i$ can be aggregated together, since extending $u_{i+1}$ does not depend on $u_i$ (according to Lemma \ref{lemma:redundant}). Based on this, we propose the following \emph{black-white vertex encoding} and \emph{aggregated embedding}. 

\subsection{Black-white Vertex Encoding \& Aggregated Embedding}

\begin{definition}[Black-White Vertex Encoding \& Aggregated Embedding]
\label{def:vertex encoding}
Given a query graph $Q$ with a matching order $O = (u_0, u_1, \dots, u_{n-1})$, each query vertex $u_i \in V(Q)$ is assigned a color $c(u_i) \in \{\text{black}, \text{white}\}$. An \emph{aggregated embedding} $M$ represents multiple embeddings and satisfies:
\begin{itemize}[leftmargin=10pt]
\item $M[u]$ is a single data vertex if $c(u)=\text{black}$;
\item $M[u]$ is a set of data vertices if $c(u)=\text{white}$.
\end{itemize}

\end{definition}

If there is no ambiguity, the ``embedding'' mentioned in the rest of this paper refers to an aggregated embedding. And we denote the black (white) backward neighbors of $u$ as $BK(u)$ ($WT(u)$).

\begin{example}
Consider again the query graph $Q$ and data graph $G$ in \figref{fig:graph_example}.
Let the matching order be $O = (u_0, \cdots, u_6)$, where $u_3$ and $u_6$ are encoded as white vertices.
When $u_0$, $u_1$, and $u_2$ are mapped to $v_0$, $v_1$, and $v_3$, respectively, $u_3$ has only one candidate $v_6$, yielding the aggregated embedding
$\{(u_0, v_0), (u_1, v_1), (u_2, v_3), (u_3, \{v_6\})\}$.
In contrast, mapping $(u_0, u_1, u_2)$ to $(v_0, v_1, v_2)$ produces
$\{(u_0, v_0), (u_1, v_1), (u_2, v_2),$ $(u_3, \{v_4, v_5, v_6\})\}$,
which compactly represents three embeddings.
\end{example}

\subsection{Four Cases of Extension}
\label{sec:four_cases}

Based on the black-white vertex encoding, we propose a common extension merging enumeration framework, as shown in \algref{algo:hybrid_enumerate}. This framework computes the extensions of $Q_{i+1}$ in the form of aggregated embeddings, based on an aggregated embedding $M$ of $Q_{i}$.
In summary, there are four cases for extending $u_i$, classified according to the encoding of $u_i$ and the encodings of its backward neighbors $N_-^O(u_i)$.

\nosection{Case 1.} $u_i$ is a black vertex, and all its backward neighbors $N_-^O(u_i)$ are black vertices.

Case 1 is equivalent to the standard extension framework described in \algref{algo:set_intersection_framework}. 
In this case, we directly compute the set of extensible vertices $R_M(u_i)$, and extend $M$ by mapping $u_i$ to each data vertex $v \in R_M(u_i)$, resulting in an embedding $M\oplus v$ of $Q_i$. 
\figref{fig:case1_eg} illustrates an example of Case 1. For illustration purposes, we treat $u_3$ as a black vertex in this example.

\begin{figure}[htbp]
	\centering
	\setlength{\abovecaptionskip}{0.2cm}
    \begin{subfigure}[c]{0.47\linewidth}
	    \includegraphics[width=\linewidth]{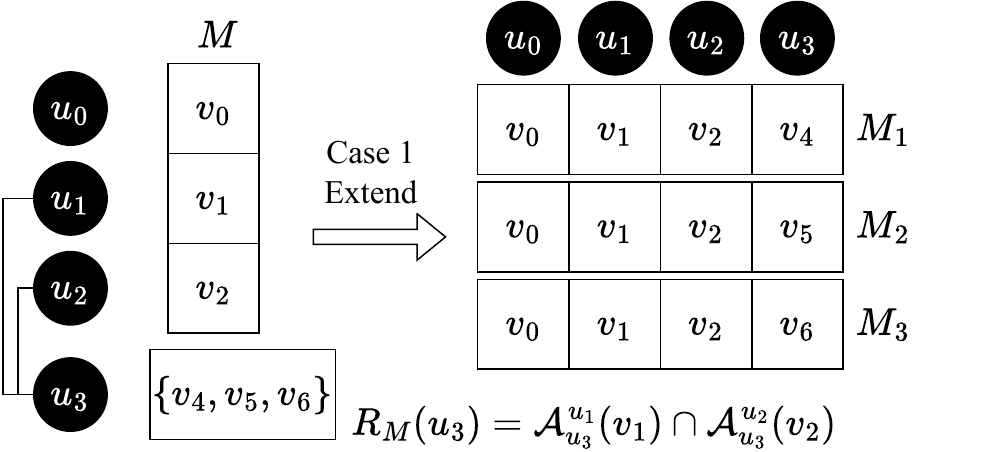}
        \caption{Case 1.}
		\label{fig:case1_eg}
	\end{subfigure}
    \hspace{0.2cm}
    \begin{subfigure}[c]{0.47\linewidth}
        \includegraphics[width=\linewidth]{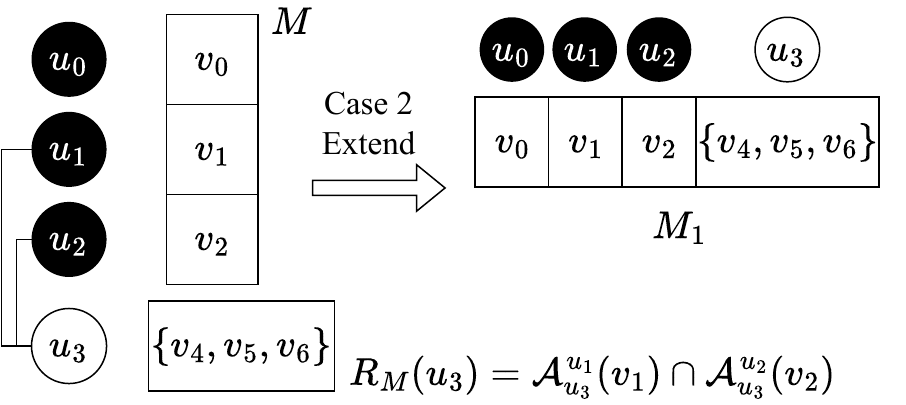}
        \caption{Case 2.}
		\label{fig:case2_eg}
        \end{subfigure}
	\caption{Illustrations of Case 1 and Case 2.}
	\label{fig:case12_eg}
\end{figure}

\nosection{Case 2.} $u_i$ is a white vertex, and all its backward neighbors $N_-^O(u_i)$ are black vertices.

White encoding of $u_i$ indicates an opportunity for extension merging. Specifically, after computing the extensible vertices $R_M(u_i)$ (same as Case 1), we aggregate them as the mapping for $u_i$ to form an aggregated embedding $M \oplus R_M(u_i)$. \figref{fig:case2_eg} illustrates this extension pattern. Note that in Case 2, the extension branches $\{M \oplus v_i \mid v_i \in R_M(u_i)\}$ are merged into a single aggregated embedding to reduce redundant extensions.

\setlength{\textfloatsep}{-5pt}    
\begin{small}
\begin{algorithm}
    \caption{Enumerate($Q$, $\mathcal{A}$, $O$, $\mathcal{M}$, $M$, $i$) (black-white enumeration framework)}
    \label{algo:hybrid_enumerate}
    \SetAlgoLined
    \DontPrintSemicolon
    \SetKwFunction{Enumerate}{Enumerate}
    
    \KwIn{The query $Q$, auxiliary structure $\mathcal{A}$, matching order $O$, result set $\mathcal{M}$, an (aggregated) embedding $M$ of $Q_{i}$, and the backtracking depth $i$.}
    \If{$i$ = $|V(Q)|$}{
        Append valid full embeddings in $M$ to $\mathcal{M}$; \textbf{return}\;
    }
    $BK(u_i) (WT(u_i)) \leftarrow$ $\{ u\mid c(u) = black\, (white), u \in N^O_-(u_i) \}$\;
    \If(\tcp*[f]{Case 1 or Case 2}){$WT(u_i) = \emptyset$}{
        $R_M(u_i)\leftarrow \cap_{u_j\in BK(u_i)}\mathcal{A}_{u_i}^{u_j}(M[u_j])$\;
        \If(\tcp*[f]{Case 1}){$c(u_i)=black$} {
            \ForEach {$v \in R_M(u_i)$}{
                \Enumerate($Q$, $\mathcal{A}$, $O$, $\mathcal{M}$, $M \oplus v$, $i+1$)\;
            }
        }
        \Else(\tcp*[f]{Case 2}){
            \Enumerate($Q$, $\mathcal{A}$, $O$, $\mathcal{M}$, $M \oplus R_M(u_i)$, $i+1$)\;
        }
    }
    \Else(\tcp*[f]{Case 3 or Case 4}){
        \If{$BK(u_i) \neq \emptyset$} {
            $R_M(u_i)\leftarrow \cap_{u_j\in BK(u_i)}\mathcal{A}_{u_i}^{u_j}(M[u_j])$\;
        }
        \Else {
            $u_j\leftarrow$ $\arg\min_{u\in WT(u_i)} |M[u]|$\;
            $R_M(u_i)\leftarrow\cup_{v \in M[u_j]} \mathcal{A}_{u_i}^{u_j}(v)$\;
        }
        \If(\tcp*[f]{Case 3}){$c(u_i)=black$} {
            \ForEach{$v\in R_M(u_i)$} {
                $M_v\leftarrow M$\;
                \ForEach{$u_j \in WT(u_i)$}{
                    $M_v[u_j] \leftarrow M[u_j] \cap \mathcal{A}_{u_j}^{u_i}(v)$\;
                }
                \If{$\forall u_j \in WT(u_i), M_v[u_j] \neq \emptyset$}{
                        \Enumerate($Q$, $\mathcal{A}$, $O$, $\mathcal{M}$, $M_v \oplus v$, $i+1$)\;
                }
            }
        }
        \Else(\tcp*[f]{Case 4}){
            $\mathcal{S} \leftarrow$decompose $WT(u_i)$'s mappings\;
            \If(\tcp*[f]{Case 4.1}){$|\mathcal{S}|  \geq|R_M(u_i)|$}{
                Same pseudocode as lines 18-23\;
            }
            \Else(\tcp*[f]{Case 4.2}){
                \ForEach{$M_t \in \mathcal{S}$}{
                    $R_{M_t}(u_i)\leftarrow \cap_{u_j\in N_-^O(u_i)}\mathcal{A}_{u_i}^{u_j}(M_t[u_j])$\; 
                    \Enumerate($Q$,$\mathcal{A}$, $O$, $\mathcal{M}$, $M_t \oplus R_{M_t}(u_i)$, $i+1$)\;
                }
            }
        }
    }
\end{algorithm}
\end{small}

\nosection{Case 3.} $u_i$ is a black vertex, and at least one of its backward neighbors is a white vertex.

As illustrated in \figref{fig:case3_naive}, $M$ is an aggregated embedding of $Q_5$ in \figref{fig:graph_example}, where the matching part $\{(u_0, v_0), (u_1, v_1), (u_2, v_2)\}$ is omitted for simplicity. The white vertex $u_3$ is mapped to three data vertices: $\{v_4, v_5, v_6\}$. When extending $u_5$, a naive solution may split $M$ into three separate embeddings for its white backward neighbor $u_3$ and perform the extension as in Case 1. However, this leads to redundant computations. By preserving aggregation instead, we can merge the resulting embeddings $M_1, M_3, M_5$ into a single $M'$, thereby facilitating more sharing opportunities for future extensions such as $u_6$.

To address this efficiently, we propose the solution illustrated in \figref{fig:case3_eg}.  
We first compute the set of extensible vertices $R_M(u_i)$, which will be detailed later.  
Then, for each vertex $v_i \in R_M(u_i)$, we prune the candidate vertices $v_j$ in the aggregated mapping $M[u_j]$ of each white backward neighbor $u_j$ of $u_i$ (lines 20-21). Specifically, if $v_i$ is not adjacent to $v_j$, $v_j$ is filtered out.

If the mapping set of any white backward neighbor becomes empty after pruning, $v_i$ is removed from $R_M(u_i)$. Otherwise, $v_i$ is appended to the shrunk aggregated embedding to form a valid extension (lines 22-23).  

\setlength{\textfloatsep}{\origtextfloatsep}

We now return to computing the extensible vertex set $R_M(u_i)$, which depends on whether $u_i$ has black backward neighbors:
\begin{itemize}[leftmargin=10pt]
    \item \textbf{If $u_i$ has at least one black backward neighbor} (lines 12-13), the extensible vertices are obtained by intersecting the adjacency sets of all such neighbors:  
    $R_M(u_i) = \bigcap_{u_j \in BK(u_i)} \mathcal{A}_{u_i}^{u_j}(M[u_j])$.
    \item \textbf{If $u_i$ has no black backward neighbors} (lines 14-16), 
    we select the white neighbor $u_j \in WT(u_i)$ with the minimum $|M[u_j]|$ and compute the extensible vertices by taking the union of adjacency sets: $R_M(u_i) = \bigcup_{v_j \in M[u_j]} \mathcal{A}_{u_i}^{u_j}(v_j)$.
\end{itemize}

\begin{figure}[htbp]
    \vspace{-0.2cm}
    \centering
    \setlength{\abovecaptionskip}{0.2cm}

    \begin{subfigure}[b]{0.85\linewidth}
        \centering
        \includegraphics[width=\linewidth]{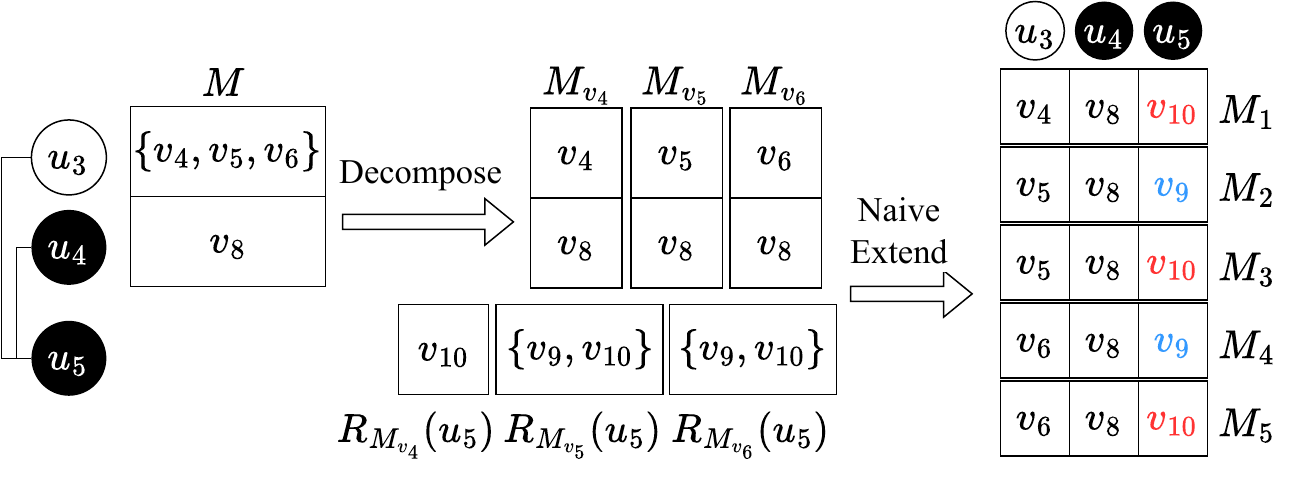}
        \captionsetup{skip=1.6pt}
        \caption{Naive solution Case 3}
        \label{fig:case3_naive}
    \end{subfigure}
    
    \vspace{0.1cm} 
    
    \begin{subfigure}[b]{0.7\linewidth}
        \centering
        \includegraphics[width=\linewidth]{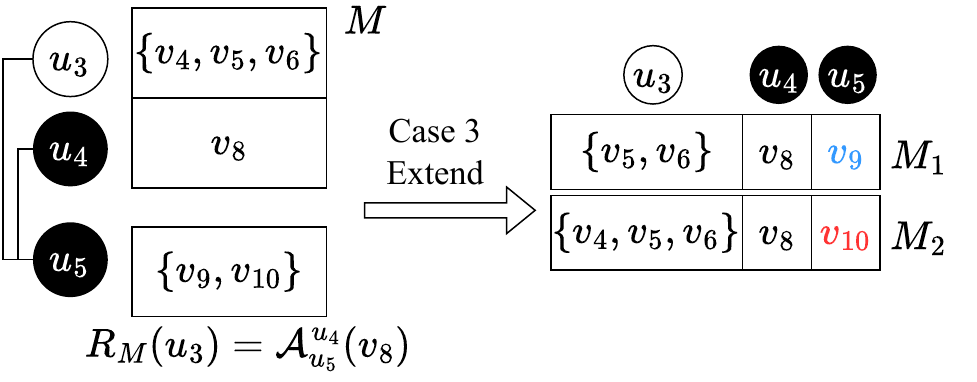}
        \captionsetup{skip=1.6pt}
        \caption{Our solution for Case 3}
        \label{fig:case3_eg}
    \end{subfigure}
    
    \caption{Illustrations of Case 3.}
    \label{fig:case3}
\end{figure}

\nosection{Case 4.} $u_i$ is a white vertex, and at least one of its backward neighbors is a white vertex.

We first compute the extensible vertices $R_M(u_i)$ in the same manner as in Case 3. The key challenge in this case lies in how to handle the aggregated embeddings involving $u_i$ and its white backward neighbors. Since both $u_i$ and some of its backward neighbors are white, we consider two alternative strategies.

One is to retain the aggregated embeddings of $u_i$’s white backward neighbors and treat $u_i$ as a black vertex, i.e., process each mapping in $R_M(u_i)$ individually. The other is to break the existing aggregated embeddings into finer-grained ones, and merge $u_i$’s mappings into the aggregation to form new composite embeddings.

In our design, the choice between these two strategies is adaptive depending on which one leads to less redundant computation. Specifically, we divide Case 4 into two subcases:

\underline{\emph{Case 4.1}} (lines 26-27). If the size of the Cartesian product of the mapping sets of $u_i$’s white backward neighbors is not less than $|R_M(u_i)|$, it is more efficient to process $u_i$ separately. Thus, we retain the aggregated embeddings of its white backward neighbors and extend each $v_i \in R_M(u_i)$ independently. This is essentially the same strategy as in Case 3.

\underline{\emph{Case 4.2}} (lines 28-31). If the size of the Cartesian product of the mapping sets of $u_i$’s white backward neighbors is smaller than $|R_M(u_i)|$, merging the mappings of $u_i$ into the aggregation may reduce redundancy. In this case, we enumerate all embeddings in the Cartesian product set $\mathcal{S}$, and for each embedding $M_t \in \mathcal{S}$, we extend it with $u_i$ following the procedure used in Case 2.

\figref{fig:case4_eg} illustrates an example of this extension strategy. Suppose the white vertex $u_i$ has two white backward neighbors, $u_j$ and $u_k$. If the product of their candidate set sizes satisfies $|R[u_j]| \times |R[u_k]| \geq |R_M(u_i)|$, we apply the same strategy as in Case 3 (upper part of \figref{fig:case4_eg}) to directly extend $u_i$. Otherwise, we decompose the candidate sets of $R[u_j]$ and $R[u_k]$, and keep $u_i$'s candidates in an aggregated form, as in Case 2 (lower part of \figref{fig:case4_eg}).

\begin{figure}
    \setlength{\abovecaptionskip}{0.23cm}
    \centering
    \includegraphics[width=0.8\linewidth]{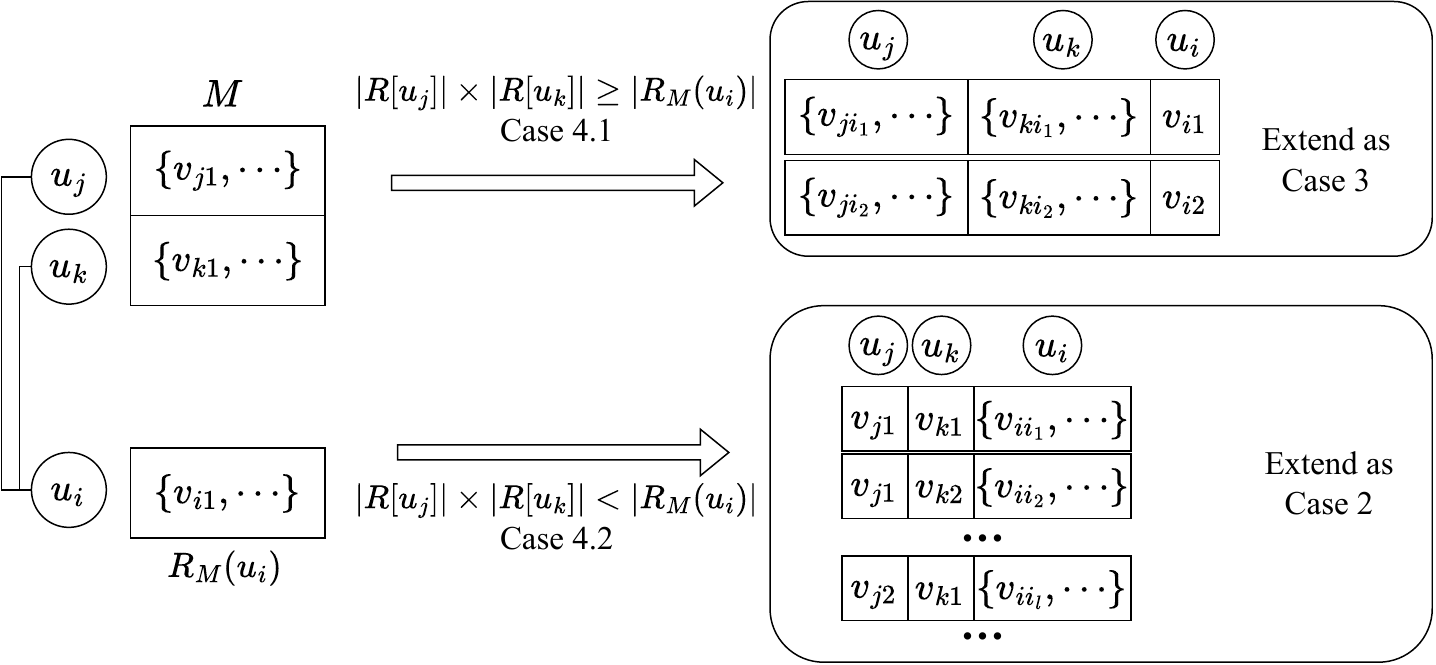}
    \caption{Illustration of Case 4.}
    \label{fig:case4_eg}
    \vspace{-0.5cm}
\end{figure}

\textbf{Discussions.} 
Compared with existing methods \cite{jin2023circinus, yang2023fast, lu2025b}, our approach does not rely on a vertex cover of the query graph and can operate with any black-white encoding of the query graph, offering greater flexibility. In particular, these methods do not handle Case 4, which may miss opportunities for merging common extensions.

\subsection{Dealing with Vertex Conflicts}
\label{sec:vertex_conflicts}

One key challenge in subgraph isomorphism is enforcing the vertex injectivity constraint, which requires that each query vertex must be mapped to a unique data vertex. 
We refer to any violation of this constraint as a \emph{vertex conflict}. 
Some existing methods, such as BSX \cite{lu2025b}, postpone conflict checking until the last level of enumeration. However, this may miss early pruning opportunities, as our experimental results demonstrate (\secref{sec:total_time_comparison}).

Due to the existence of white vertices in our enumeration framework, our approach to detecting vertex conflicts differs slightly, and is not explicitly shown in \algref{algo:hybrid_enumerate}. 
In our solution, we only record the mappings of black vertices, as well as the mappings of white vertices when they are mapped to a single data vertex. 
This includes two situations: (i) $u_i$ is deterministically mapped when extended under Case 3 or Case 4; and (ii) under Case 4, a white backward neighbor $u_k$ of $u_i$ becomes deterministically mapped when its candidate set is reduced to a single vertex.
Such deterministically mapped vertices are included in conflict checking to preserve correctness and maximize pruning opportunities.

At the final level of enumeration, we obtain an aggregated embedding $M$ in which every query vertex has been assigned a mapping. Since each white vertex may map to multiple candidate data vertices, we generate all possible full embeddings by taking the Cartesian product of their mappings.
Each resulting embedding is then checked against the vertex injectivity constraint, and only valid ones are added to the result set $\mathcal{M}$ (lines 1-2 in \algref{algo:hybrid_enumerate}).

\section{Common Extension Reusing}
\label{sec:reuse}

\texttt{CEM} can share the extension computation in an aggregated embedding, but it still misses other optimization opportunities, as illustrated in the following example. 

\begin{example}\label{example:cer} In \figref{fig:parent_vertex}, we extend $u_4$ in the query graph given in Figure \ref{fig:graph_example}. Because $u_4$'s backward neighbors are $u_0$ and $u_1$, \lemref{lemma:redundant} tells us that the extension of $u_4$ only depends on $u_0$ and $u_1$, while $u_2$ and $u_3$ are \emph{irrelevant}. Using \texttt{CEM}, assume that we have generated two partial embeddings $M_1$ and $M_2$ for subquery $Q_3$, which have the same mappings of $u_0$ and $u_1$, but different in $u_2$ and $u_3$. When extending $u_4$, we need to extend $M_1$ and $M_2$, respectively. However, the extension of $u_4$ is the same for both $M_1$ and $M_2$. Thus, a desirable solution should avoid such duplicated extension, which cannot be achieved by \texttt{CEM}. 
\end{example}

To address the above problem, we propose a \emph{backward-looking} \underline{c}ommon \underline{e}xtension \underline{r}eusing technique, called \texttt{CER}, which uses the historical results already generated to avoid duplicated extension.

\begin{figure}[htbp]
	\centering
	\includegraphics[width=0.78\linewidth]{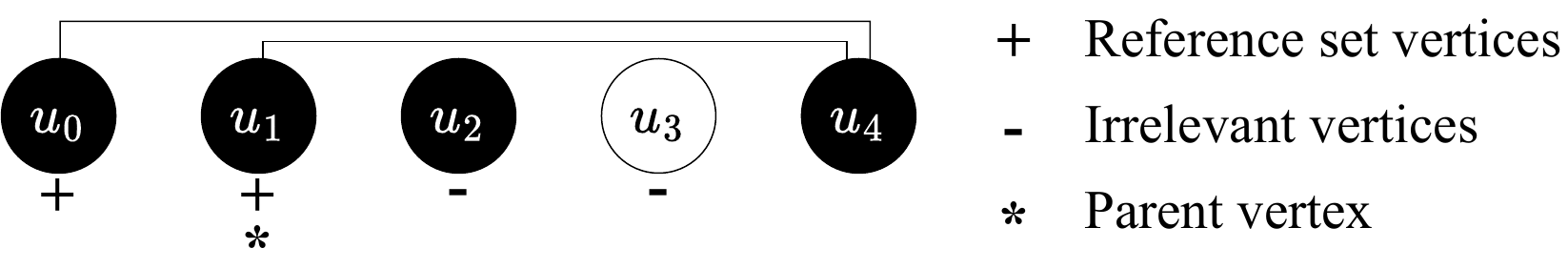}
	\caption{The reference set and parent vertex of $u_4$}
	\label{fig:parent_vertex}
    \vspace{-0.4cm}
\end{figure}

\subsection{Reference Set \& Brother Embeddings}

We first give the definitions of \emph{reference set} and \emph{brother embeddings}, which are main concepts to introduce \texttt{CER}. The idea behind them is based on a relaxed version of Lemma \ref{lemma:redundant}.

\begin{definition}[Reference Set]
\label{def:reference_set} 
Assume that the next query vertex to be extended is $u_i$. The \emph{reference set} of $u_i$ is a subset of $\{u_0, \cdots, u_{i-1}\}$, denoted as $RS(u_i)$, includes two components: 
\begin{itemize}[leftmargin=10pt]
    \item The closure of $u_i$'s backward neighbors, denoted as $Anc(u_i)$, which contains $N_-^O(u_i)$ and all their backward neighbors recursively;
    \item All vertices $u_k$ (with $k<i$) that are adjacent to at least one white backward neighbor of $u_i$.
\end{itemize}
Formally, the reference set $RS(u_i)$ is defined as:
\begin{equation}
\resizebox{0.926\hsize}{!}{$
RS(u_i) = Anc(u_i) \cup \{ u_k \mid\exists u_j \in WT(u_i),\ e(u_j, u_k) \in E(Q),\ k < i\}.
$}
\end{equation}
\end{definition}

Intuitively, the reference set $RS(u_i)$ contains the query vertices whose mappings affect the extensible set $R_M(u_i)$. Besides the direct backward neighbors of $u_i$, it also includes vertices $u_k$ ($k < i$) that are adjacent to any white backward neighbor $u_j \in WT(u_i)$. This is because $u_k$ is extended using Case 3 or Case 4, which could prune the mappings of $u_j$, and thereby influence the extension of $u_i$.

\begin{definition}[Brother Embeddings]
\label{def:brother_embeddings}
Assume that $M_1$ and $M_2$ are two partial embeddings of the subquery $Q_i$, and the next query vertex to be extended is $u_i$. 
We say that $M_1$ and $M_2$ are \emph{brother embeddings} if they assign the same mapping to every vertex in the reference set of $u_i$, i.e., $M_1(u) = M_2(u)$ for all $u \in RS(u_i)$.
\end{definition}

\begin{example}
Recall \exgref{ex:motivation2}, $M_1$ and $M_2$ are two partial embeddings of the subquery $Q_4$. 
As shown in \figref{fig:parent_vertex}, the next query vertex $u_4$ is adjacent to $u_0$ and $u_1$, both of which are black vertices. 
Therefore, the reference set of $u_4$ is $RS(u_4) = \{u_0, u_1\}$. In this case, $M_1$ and $M_2$ assign the same mappings to the vertices in $RS(u_4)$, i.e., $M_1[u_0] = M_2[u_0] = v_0$ and $M_1[u_1] = M_2[u_1] = v_1$. Thus, $M_1$ and $M_2$ are brother embeddings.
\end{example}

\subsection{CER Strategy for Extension Reusing}
\label{sec:cer_for_reusing}

We now propose the \texttt{CER} strategy. 
Given a query graph $Q$, for each query vertex $u_i$, we define its \emph{parent vertex} as the vertex $u_k \in RS(u_i)$ with the largest index in the matching order $O$ (e.g., in \figref{fig:parent_vertex}, $u_1$ is the parent vertex of $u_4$). Correspondingly, $u_i$ is referred to as a \emph{child vertex} of $u_k$.

If $u_k$ is not the immediate predecessor of $u_i$ in $O$, i.e., $k < i - 1$, we set a flag $u_i.f = \texttt{true}$, record its parent as $u_i.p = u_k$, and register $u_i$ as a child of $u_k$ via $u_k.child \leftarrow u_k.child \cup \{u_i\}$. 
Otherwise, we set $u_i.f = \texttt{false}$. 
The \texttt{CER} strategy is applied only to vertices $u_i$ with $u_i.f = \texttt{true}$. 
We design a buffer structure, called the \emph{Common Extension Buffer} (CEB), to cache reusable extension results.

\begin{definition}[Common Extension Buffer (CEB)]
\label{def:ceb}
Given a query vertex $u_i$ with $u_i.f = \texttt{true}$, its CEB consists of two components:
$CEB(u_i).g$, a boolean flag indicating whether the buffer is valid (initialized to \texttt{false});
and $CEB(u_i).b$, which stores reusable local extensions of $u_i$.
\end{definition}

We now describe how to use the CEB during the \texttt{CER} procedure. 
Suppose we are processing a partial embedding $M$ and the next query vertex to be extended is $u_i$, where $u_i.f = \texttt{true}$.

\begin{enumerate}[leftmargin=14pt]
    \item When $CEB(u_i).g$ is \texttt{false}:  
    This indicates the \emph{first time} $u_i$ is extended among $M$'s brother embeddings.  
    We update $CEB(u_i).b$ with the current extensions and set $CEB(u_i).g = \texttt{true}$ after processing.  
    The contents pushed into the buffer depend on the extension strategy, as discussed in \secref{sec:merge}:
    \begin{itemize}[leftmargin=10pt]
        \item Case 1 \& 2: Store $R_M(u_i)$ directly into $CEB(u_i).b$.
        \item Case 3: Store $\{M_v \oplus v \mid v \in R_M(u_i)\}$, where $M_v$ is the shrunk embedding derived from $M$ associated with $v$ (see lines 19-21 in \algref{algo:hybrid_enumerate}).
        \item Case 4: 
        For Case 4.1, store $\{M_v \oplus v \mid v \in R_M(u_i)\}$ as in Case 3;  
        For Case 4.2, store $\{M_t \oplus R_{M_t}(u_i) \mid M_t \in \mathcal{S}\}$, where $\mathcal{S}$ is the Cartesian product of decomposed candidates, and $R_{M_t}(u_i)$ is the extensible set of $u_i$ under $M_t$.
    \end{itemize}

    \item When $CEB(u_i).g$ is \texttt{true}: This means that valid extensions have already been cached. We reuse $CEB(u_i).b$ to directly extend $M$ based on $u_i$'s different extension cases, avoiding redundant computation.
    \item When backtracking from the level of $u_i$ in the search tree, we reset the CEB flags of all child vertices of $u_i$, i.e., we set $CEB(u_j).g = \texttt{false}$ for every $u_j \in u_i.child$.
\end{enumerate}

\begin{figure}[htbp]
	\centering
	\includegraphics[width=0.57\linewidth]{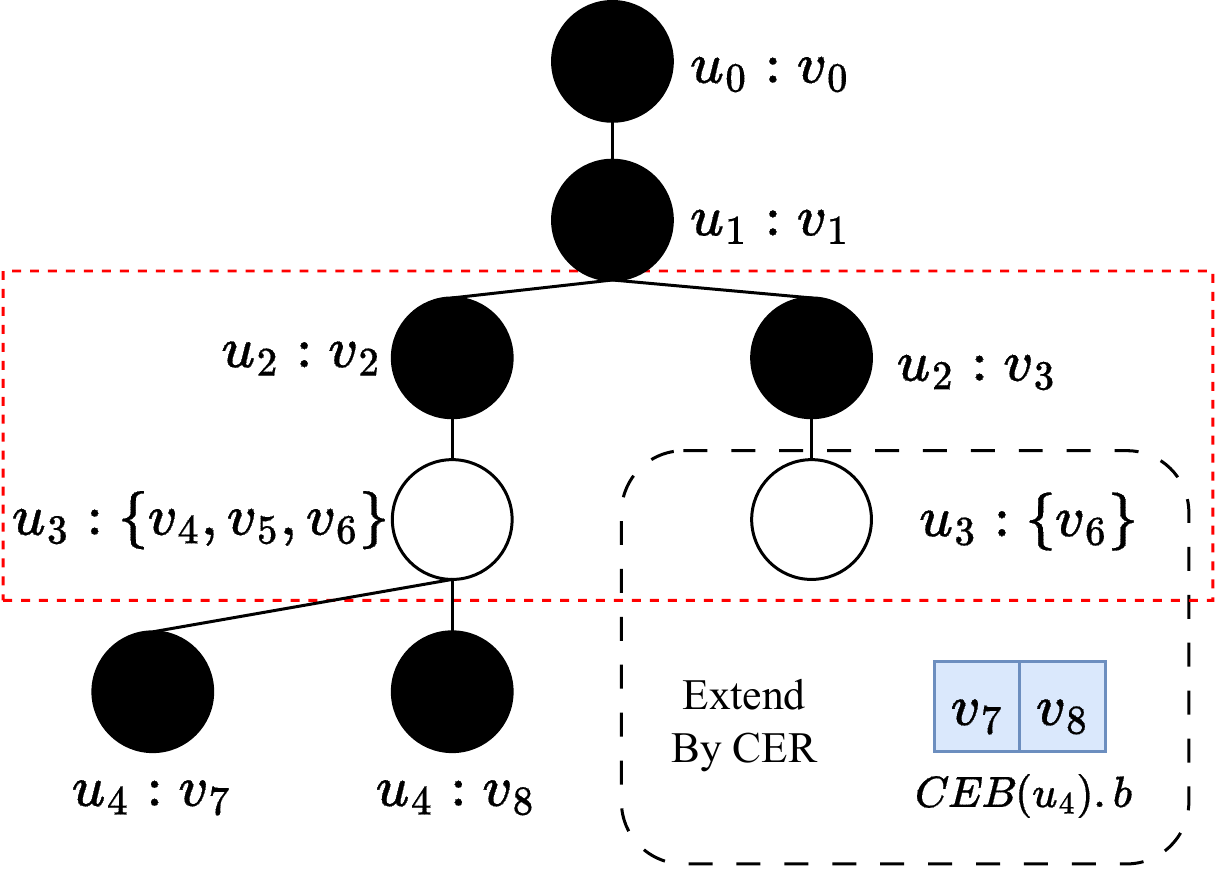}
	\caption{Using CEBs to eliminate duplicated computation}
    \label{fig:matching_tree}
\end{figure}

\begin{example}
Referring to the query and data graphs in \figref{fig:graph_example}, the search tree from $u_0$ to $u_4$ rooted at $v_0$ is illustrated in \figref{fig:matching_tree}. 
Assume that the enumeration has just backtracked to the partial embedding $M = \{(u_0, v_0), (u_1, v_1), (u_2, v_2), (u_3, \{v_4, v_5, v_6\})\}$ after the mappings $\{v_7, v_8\}$ for $u_4$ are cached in $CEB(u_4).b$. Since $u_1$ is the parent of $u_4$, all partial embeddings within the red dashed box represent the brother embeddings of $M$. Consequently, the embedding $\{(u_0, v_0), (u_1, v_1), (u_2, v_3), (u_3, \{v_6\})\}$ can be extended directly by utilizing the $CEB(u_4).b$. This $CEB$ of $u_4$ remains valid until the mapping of $u_1$ changes from $v_1$.
\end{example}

\textbf{Remark.} 
We enable \texttt{CER} only for query vertices $u_i$ whose parent vertex is not $u_{i-1}$. 
The reason is that if $u_i$'s parent is exactly $u_{i-1}$, then each time we backtrack to the depth of $u_i$, the mappings of vertices in $RS(u_i)$ may have changed due to the rematching of $u_{i-1}$. 
As a result, the previous extensible vertices of $u_i$ are no longer valid, making extension reuse ineffective.

\section{Optimizations and Extensions}
\label{sec:opt}

This section presents several key optimizations in \method, including two effective pruning techniques, matching order selection, and an encoding strategy. We further describe how \method can be extended to support directed and edge-labeled graphs.

\subsection{Pruning Techniques}
To effectively identify unpromising search branches that cannot produce any valid full embedding during enumeration, we additionally propose two pruning techniques.

\subsubsection{Contained Vertex Pruning}
\label{sec:contained_vertex_pruning}
\begin{definition}[Contained Vertex Set]
Suppose the matching order is $O=(u_0, \cdots, u_{n-1})$, and $u_i$ and $u_j$ is two vertices with the same label, and $i<j$. If the backward neighbor set of $u_i$ is $u_j$'s subset, i.e., $N_-^O(u_i)\subseteq N_-^O(u_j)$, then we say that $u_j$ is contained by $u_i$ under the matching order $O$. The set of query vertices that contained by $u_i$ under the matching order is called \textit{contained vertex set} of $u_i$, denoted by $Con(u_i)$.
\end{definition}

Then we have the following pruning rule,
\ifbool{fullversion}{whose proof can be found in Appendix \ref{proof:pruning}.}{\red{whose proof can be found in the appendix of the full paper \cite{codes}.}}

\begin{lemma}[Contained Vertex Pruning]
\label{lem:contained_vertex_pruning}
During the enumeration process of extending $u_i$ under the partial embedding $M$, if $|R_M(u_i)| < |Con(u_i)|$, then the search branch rooted at $M$ can be safely pruned.
\end{lemma}

\subsubsection{Extended Failing Set Pruning}
\label{sec:failing_set_pruning}

The failing set is an effective backjumping-based pruning technique, originally proposed in DAF \cite{han2019efficient}, and has been adopted by subsequent methods \cite{sun2020rapidmatch, kim2021versatile, choi2023bice}.

\begin{definition}[Failing Set]
Let $M$ be a partial embedding. A subset of query vertices $F_M$ is called a \emph{failing set} of $M$ if no valid full embedding can be extended from the restriction of $M$ to $F_M$, denoted as $M[F_M]$.
\end{definition}

\begin{lemma}[Failing Set Pruning]
Suppose $M$ is a partial embedding where the last matching is $(u_i, M[u_i])$, and $F_M$ is a non-empty failing set of $M$. If $u_i \notin F_M$, then all sibling branches of $M$ in the search tree can be safely pruned.
\end{lemma}

Intuitively, the failing set $F_M$ indicates those query vertices whose current mappings prevent the enumeration from producing any valid embeddings. If none of the mappings in $F_M$ are changed, subsequent search will inevitably fail. Building on this idea, we propose an \emph{extended failing set pruning} technique under the black-white vertex enumeration framework.
Specifically, the failing set of a partial embedding $M$ (where $u_i$ is the last mapped vertex in $M$) is computed as follows:
\begin{enumerate}[leftmargin=14pt]
    \item If $M$ is a leaf node in the search tree.
    \begin{itemize}[leftmargin=2pt]
        \item $M$ belongs to \textit{complete embedding class} if $M$ contains valid matchings. Then $F_M = \emptyset$.
        \item $M$ belongs to \textit{insufficient candidate class} if $|R_M(u_i)| < Con(u_i)$. Then $F_M = RS(u_i)$, corresponding to contained vertex pruning.
        \item $M$ belongs to \textit{vertex conflict class} if the mapping of $u_i$ conflicts with that of some $u_j$ where $j<i$. Then $F_M = Con(u_i)\cup Con(Tr(u_j))$, where $Tr(u_j)$ denotes the query vertex that restricts the candidate set of $u_j$ to a singleton. Specifically, if $u_j$ is a black vertex, or a white vertex with $|R_M(u_j)| = 1$, then we set $Tr(u_j) = u_j$. Otherwise, if $u_j$ is a white backward neighbor of some $u_k$ (where $k < i$), and $u_k$ is the first vertex in the matching order that reduces the mapping set of $u_j$ to a single vertex, then we set $Tr(u_j) = u_k$.
    \end{itemize}
    \item If $M$ is an internal node in the search tree, and $M_1, \dots, M_k$ are all its children.
    \begin{itemize}[leftmargin=2pt]
        \item First, if there exists a child $M_l$ such that $F_{M_l} = \emptyset$ (i.e., $M_l$ leads to a valid embedding), then $F_M = \emptyset$.
        \item Else if there exists some $M_l$ such that $u_i \notin F_{M_l}$, then $F_M = F_{M_l}$.
        \item Otherwise, $F_M = \bigcup_{l=1}^k F_{M_l}$.
    \end{itemize}
\end{enumerate}

\nosection{Remark.} Compared with the original failing set proposed in DAF \cite{han2019efficient}, our solution introduces two main extensions. First, we incorporate the contained vertex pruning rule, which enables earlier detection of failure. Second, our approach is compatible with the black-white vertex enumeration framework, where a batch of data vertices may be mapped to a single white query vertex.

\subsection{Matching Order Selection}
\label{sec:matching_ordering}
We decouple the matching order from the encoding strategy to reduce complexity. Our goal in matching order selection is to minimize the size of intermediate results during backtracking, thereby reducing the search space.

After the filtering phase, each query vertex $u$ is associated with a candidate set $C(u)$. We prioritize vertices with smaller candidate sizes and stronger connectivity to the current partial order. Specifically, the first vertex in the matching order is selected as:
\begin{equation}
u_0 = \arg\min_{u\in V(Q)} \frac{|C(u)|}{d(u)}.
\end{equation}
Given a partial matching order $O_i = (u_0, u_1, \cdots, u_{i-1})$, denote $N(O_i) = \cup_{u\in O_i}N(u)$, and the next vertex is chosen by:
\begin{equation}
u_i = \arg\min_{u\in N(O_i)\setminus O_i} \frac{|C(u)|}{|N(u)\cap O_i|}.
\end{equation}

\subsection{Encoding Strategy}
\label{sec:encoding_strategy}
Given a matching order $O$, \method searches for embeddings by considering the encoding types (black or white) of the current query vertex $u_i$ and its backward neighbors $N_-^O(u_i)$. These encoding choices influence both pruning strategies and the degree to which intermediate computations can be shared.
An effective encoding for a query vertex $u$ should balance the following factors:
\begin{itemize}[leftmargin=10pt]
\item A large number of forward neighbors or forward-neighbor candidates increases the case 3 and 4 situations, resulting in redundant splitting costs.
\item Encoding $u$ as white is less beneficial if it has many white backward neighbors, since fewer constraints are enforced. In contrast, more black backward neighbors help confirm the matching of $u$, but reduce the possibility of computation sharing.
\item Vertices with common labels are likely to cause conflicts during backtracking, and assigning them as white vertices would undermine the algorithm's ability to find such vertex conflict during intermediate search.
\item A large candidate size for $u$ suggests more potential for computation reuse among its candidates.
\end{itemize}

To jointly consider these factors, we propose a cost model that compares the white risk $WR(u)$ and the black risk $BR(u)$:
\begin{align}
WR(u) &= (1 + \sum_{u'\in N_+^O(u)} |C(u')|) \times |V(Q, L(u))| \times |WT(u)|; \\
BR(u) &= |C(u)| \times |BK(u)|,
\end{align}
where $|V(Q, L(u))|$ is the number of query vertices of label $u$. We encode $u$ as white if $WR(u) < BR(u)$, and as black otherwise.

\subsection{Extensions}
\label{sec:extensions}

Although \method is described under the undirected, vertex-labeled graph model, it can be readily extended to directed and edge-labeled graphs. Only two components require modification:
\begin{itemize}[leftmargin=10pt]
\item \textbf{Filtering stage.} Candidate edges must respect both the direction and the label of query edges during the filtering process.
\item \textbf{Definition of containing vertex.} During enumeration, edge directions must be considered for containing vertices. For two vertices $u_i$ and $u_j$ ($i < j$), if both the in- and out-backward neighbors of $u_i$ are subsets of those of $u_j$, then $u_j$ is said to be contained by $u_i$, enabling containing pruning.
\end{itemize}

\section{Experiments}
\label{sec:experiment}

\subsection{Experimental Setup}
\label{sec:exp_setup}

\subsubsection{Datasets}

We conduct our experiments on eight real-world datasets: Yeast, Human, HPRD, WordNet, DBLP, EU2005, YouTube, and Patents. These datasets are widely used in previous studies \cite{han2019efficient, sun2020rapidmatch, arai2023gup, kim2021versatile, choi2023bice} and surveys \cite{sun2020memory, zhang2024comprehensive}, and they span a variety of domains, including biology (Yeast, Human, HPRD), lexical semantics (WordNet), social networks (DBLP, YouTube), the web (EU2005), and citation networks (Patents). They differ in scale and complexity (e.g., topology and density), and their detailed statistics are summarized in \tabref{tab:dataset}.

\begin{small}
\begin{table}[t]
    \caption{Datasets statistics}
    \label{tab:dataset}
\begin{tabular}{ccccc}
    \toprule
        Dataset & $\lvert V \rvert$ & $\lvert E \rvert$ & $\lvert \Sigma \rvert$ & average degree \\
        \midrule
        Yeast & 3,112 & 12,519 & 71 & 8.0\\
        Human & 4,674 & 86,282 & 44 & 36.9\\
        HPRD & 9,460 & 34,998 & 307 & 7.4\\
        WordNet & 76,853 & 120,399 & 5 & 3.1 \\
        DBLP & 317,080 & 1,049,866 & 15 & 6.6\\
        EU2005 & 862,664 & 16,138,468 & 40 & 37.4\\
        YouTube & 1,134,890 & 2,987,624 & 25 & 5.3\\
        Patents & 3,774,768 & 16,518,947 & 20 & 8.8\\
    \bottomrule
\end{tabular}
\end{table}
\end{small}

\subsubsection{Queries}
\nop{Following prior work \cite{sun2020memory, han2019efficient, arai2023gup, zhang2024comprehensive}, we generate query graphs by performing random walks on data graphs and extracting the induced subgraphs. This ensures that every query graph have at least one valid embedding. For Human and WordNet (challenging due to topology denseness and label sparseness, respectively), we use query sizes $\{4, 8, 12, 16, 20\}$, and for the other six datasets, we use $\{4, 8, 16, 24, 32\}$. 
For each query size $n$, we denote the set as $\mathcal{Q}_n$ and further divide it into sparse ($\mathcal{Q}_{nS}$, average degree < 3) and dense ($\mathcal{Q}_{nD}$, otherwise) subsets when $n \geq 4$. Each dataset includes 200 queries for each of $\mathcal{Q}_{nS}$ and $\mathcal{Q}_{nD}$, and 200 queries for $\mathcal{Q}_4$, resulting in 1800 queries per dataset.}

We use 10,000 queries for each dataset in our experiments, covering a wide range of sizes and densities.
Following prior work \cite{sun2020memory, han2019efficient, arai2023gup, zhang2024comprehensive}, we generate query graphs by performing random walks on the data graphs and extracting the induced subgraphs. This procedure guarantees that every query graph has at least one valid embedding.
For Human and WordNet (challenging due to topology denseness and label sparseness, respectively), we use query sizes of $\{4, 8, 12, 16, 20\}$, and for the other six datasets, we use $\{4, 8, 16, 24, 32\}$.
For each query size $n$, we denote the corresponding query set as $\mathcal{Q}_n$, which contains 2,000 queries. Thus, each dataset comprises a total of $2000 \times 5 = 10000$ query graphs.
When $n > 4$, $\mathcal{Q}_n$ is further divided into a sparse subset $\mathcal{Q}_{nS}$ (average degree < 3) and a dense subset $\mathcal{Q}_{nD}$ (otherwise), i.e., $\mathcal{Q}_n = \mathcal{Q}_{nS} \cup \mathcal{Q}_{nD}$, with 1,000 queries in each subset.

\subsubsection{Compared Methods}
For performance comparison, we evaluate our \method algorithm with six state-of-the-art subgraph matching algorithms: DAF \cite{han2019efficient}, RM \cite{sun2020rapidmatch}, VEQ \cite{kim2021versatile}, GuP \cite{arai2023gup},  BICE \cite{choi2023bice} and BSX \cite{lu2025b}. 
All the source codes of the compared algorithms come from the original authors. 

\subsubsection{Experiment Environment}
Except for GuP, which is implemented in Rust, all other algorithms are implemented in C++.
We compile all C++ source code using g++ 13.1.0 with the -O2 flag. All experiments are conducted on an Ubuntu Linux server with an Intel Xeon Gold 6326 2.90 GHz CPU and 256 GB of RAM.
Each experiment is repeated three times, and the mean values are reported.

\subsubsection{Metrics}
\label{sec:exp_metrics}
To evaluate performance, we measure the elapsed time in milliseconds (ms) for processing each query, which consists of two components: preprocessing time and enumeration time. Specifically, the preprocessing time includes filtering and ordering steps (lines 1-2 of \algref{algo:generic_framework}), and for \method, it also includes the encoding step. The enumeration time refers to the duration of the backtracking search (line 4 of \algref{algo:generic_framework}). Except for the experiment in \secref{sec:limit_number_enumeration_time_study}, we terminate the algorithm after finding the first $10^6$ results as default.  To ensure queries complete within a reasonable time, we impose a timeout limit of 6 minutes; if a query exceeds this limit, its elapsed time is recorded as 6 minutes.

\subsection{Comparison with Existing Methods}
\label{sec:comp_existing_methods}

\subsubsection{Total Time Comparison}
\label{sec:total_time_comparison}

\figref{fig:total_time_comparison} presents the total query processing time of all compared methods.
For clarity, we separate the enumeration time (bottom bars with hatching) from the preprocessing time (top solid bars).
Except for HPRD, \method\ consistently outperforms all other methods, achieving a speedup of 1.39× to 9.80× over the second fastest method.
This improvement is mainly achieved in the enumeration stage, owing to the extension-reduction techniques of \method, and it becomes more pronounced for queries with large result sizes (see \secref{sec:limit_number_enumeration_time_study}).

HPRD is a relatively simple dataset with a large number of labels.
In such cases, most of the runtime is dominated by filtering, and DAF performs filtering more efficiently on HPRD, resulting in a shorter total time than \method.
However, as shown in \secref{sec:enumeration_time_comparison}, \method\ still surpasses DAF in enumeration on HPRD.

Some methods involve additional preprocessing steps.
For example, BICE clusters vertices that share the same neighbors after filtering, and GuP calculates guards for all candidates.
These extra steps can lead to poor total performance on certain datasets (e.g., BICE on WorNet, GuP on Patents).

\method also outperforms BSX in terms of efficiency in our experiments.
We attribute this performance gap to two main factors.
First, \method incorporates highly effective pruning strategies, such as the failing set pruning, which are absent in BSX.
Second, BSX only guarantees homomorphism during the batch search and postpones the injectivity check to the final enumeration stage.
Consequently, it incurs unnecessary overhead by processing many intermediate candidates that appear feasible but ultimately fail to satisfy the injectivity constraint.

\begin{figure}
    \setlength{\abovecaptionskip}{0.2cm}
    \centering
    \includegraphics[width=0.98\linewidth]{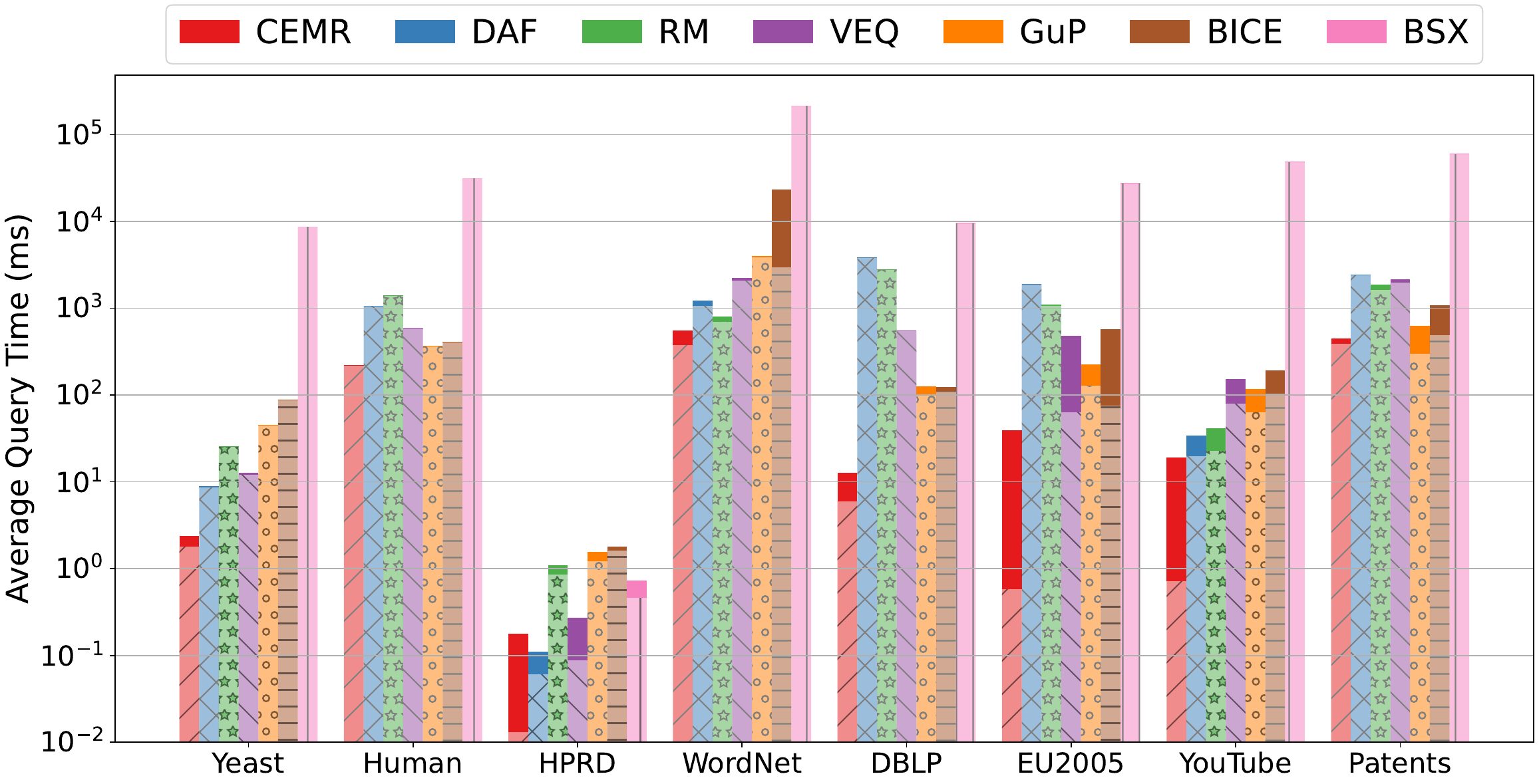}
    \caption{Average query time comparison. The bottom bars with hatching represent enumeration time, while the solid bars above represent preprocessing times.}
    \label{fig:total_time_comparison}
    \vspace{-0.15cm}
\end{figure}

\subsubsection{Enumeration Time Comparison}
\label{sec:enumeration_time_comparison}

Since \method is primarily designed to accelerate the enumeration phase, we conduct a detailed evaluation of its enumeration performance.
From the results in \figref{fig:total_time_comparison}, \method processes queries faster than other methods in most cases (with the exception of Patents, due to one additional timeout query, as discussed in \secref{sec:unsolved_ana_exp}). Compared with the second-best method, \method achieves a speedup ranging from 1.67x to 108.52x. As will be shown in \secref{sec:ablation_studies}, this advantage largely stems from the proposed \texttt{CEM} and \texttt{CER} techniques.

We further analyze enumeration time under different query sizes. The results in \figref{fig:enum_time_comparison_diff_size} show a consistent trend across datasets: enumeration time increases with query size for all methods, and \method generally outperforms its competitors.
A notable case is EU2005, where the correlation between query size and enumeration time does not strictly hold. In this dataset, larger query graphs, after filtering, almost always yield matches with large result sets. Consequently, \method can quickly generate embeddings and enumerate a substantial number of results by redundant extension elimination, leading to comparable or even shorter enumeration times for larger queries compared with smaller ones.

\begin{figure*}[!t]
    \centering
    \begin{subfigure}[c]{0.9\linewidth}
        \centering
        \includegraphics[width=\linewidth]{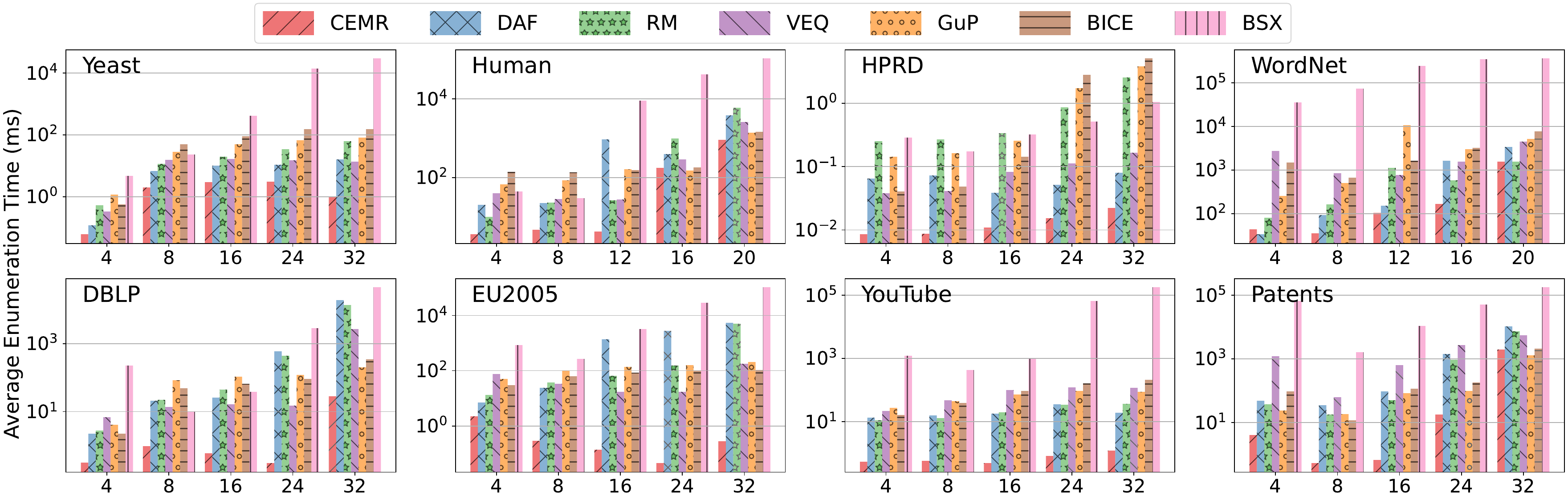}
        \captionsetup{skip=2pt}
        \caption{Enumeration time comparison across different query sizes.}
        \label{fig:enum_time_comparison_diff_size}
    \end{subfigure}
    \vspace{0.25cm}
    \begin{subfigure}[c]{0.9\linewidth}
        \centering
        \includegraphics[width=\linewidth]{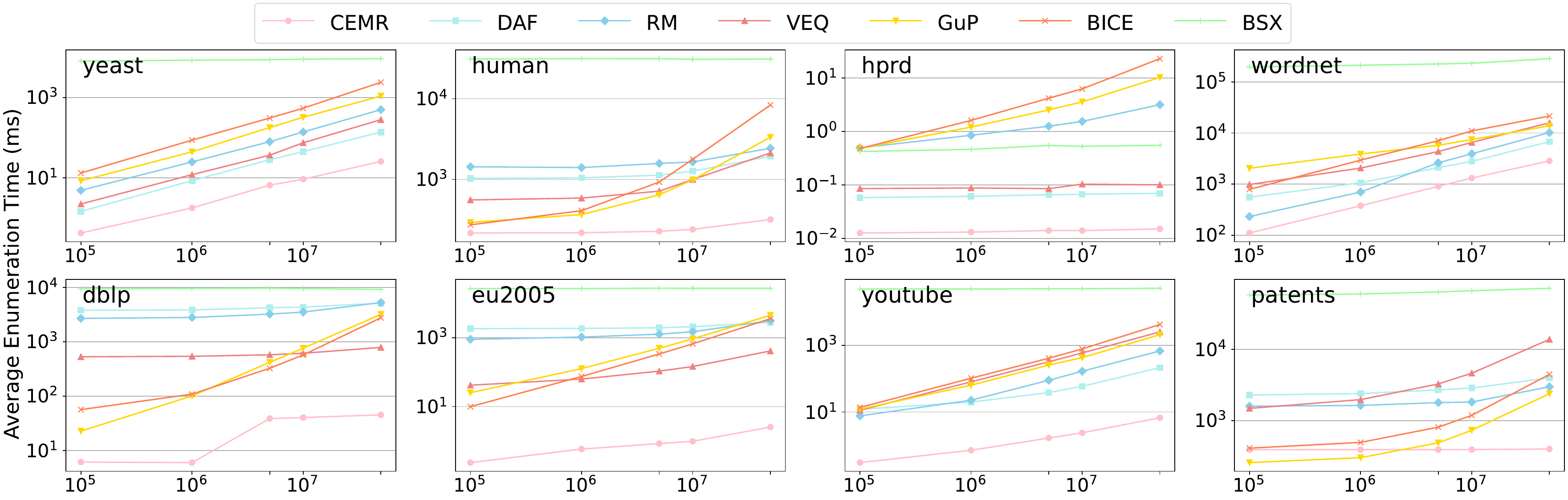}
        \captionsetup{skip=2pt}
        \caption{Enumeration time comparison under different result size limits.}
        \label{fig:enum_time_comparison_diff_ret}
    \end{subfigure}
    \vspace{-0.4cm}
    \caption{Detailed enumeration time comparison results. (a) Varying query sizes; (b) Varying result size limits.}
\end{figure*}

\begin{figure}
    \setlength{\abovecaptionskip}{0.1cm}
    \centering
    \begin{subfigure}[c]{0.45\linewidth}
        \centering
        \includegraphics[width=\linewidth]{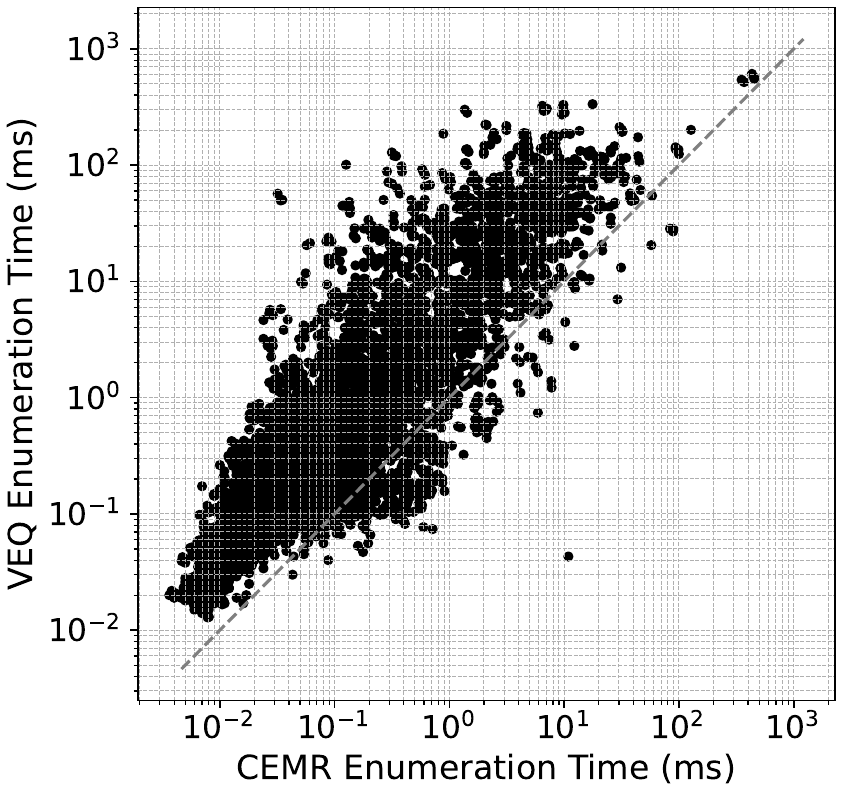}
        \caption{Yeast}
        \label{fig:yeast_scatter}
    \end{subfigure}
    \hspace{0.1cm}
    \begin{subfigure}[c]{0.45\linewidth}
        \centering
        \includegraphics[width=\linewidth]{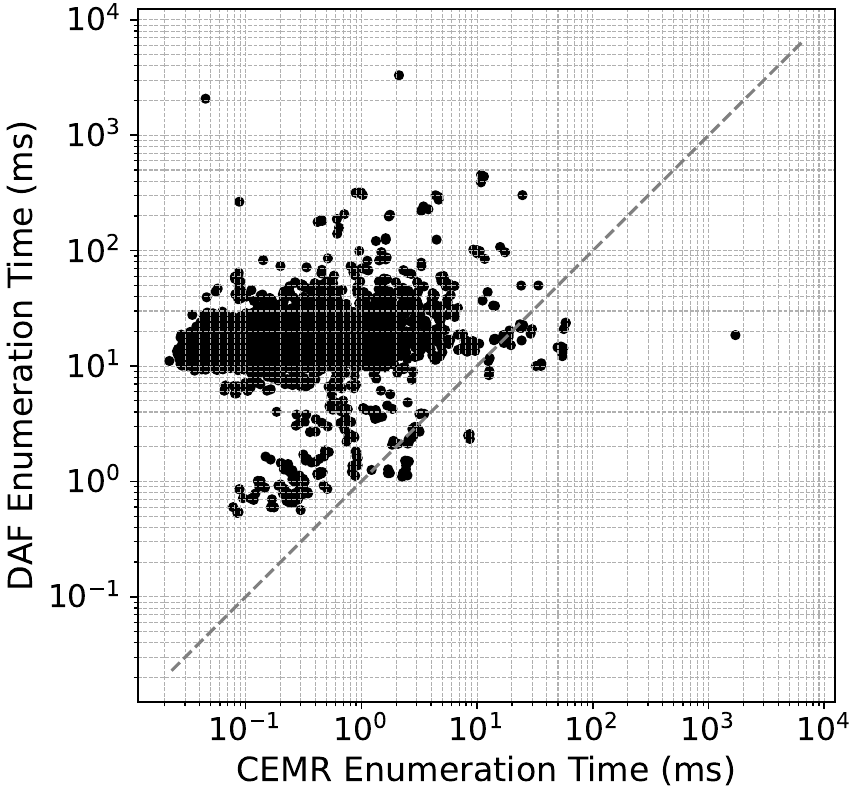}
        \caption{YouTube}
        \label{fig:youtube_scatter}
    \end{subfigure}
    \caption{Enumeration time comparison on Yeast and YouTube against VEQ and DAF, respectively.}
    \label{fig:yeast_youtube_enumeration}
    \vspace{-0.5cm}
\end{figure}

To better understand the performance of \method, we compare the enumeration time of each query with that of the second-best method based on the average enumeration time on the largest-size queries in the corresponding dataset.
The scatter plots in \figref{fig:yeast_youtube_enumeration} present the results for Yeast and YouTube, where the second-best methods are VEQ and DAF, respectively. 
In both cases, all queries are processed within the time limit by both methods. 
Most points lie above the diagonal line, indicating that \method achieves shorter enumeration times on the vast majority of queries. The improvement is particularly pronounced for queries with large result sets, where \method benefits from rapid embedding generation. For smaller queries, the advantage is smaller but still consistent, suggesting that \method not only scales well with query size but also maintains low per-query overhead. These results confirm that the performance gain of \method is not limited to a few extreme cases but is robust across different query characteristics.

\subsubsection{Unsolved Queries Number Comparison}
\label{sec:unsolved_ana_exp}

\begin{small}
\begin{table}[ht]
    \caption{Number of unsolved queries}
    \label{tab:unsolved}
  \begin{tabular}{ | c | c | c | c | c | c | } 
    \hline
    Methods & Human & WordNet & DBLP & EU2005 & Patents \\
    \hline
        DAF & 25 & 16 & 86 & 50 & 56 \\ \hline
        RM & 33 & 4 & 58 & 22 & 34 \\ \hline
        VEQ & 14 & 9 & 13 & \textbf{0} & 22 \\ \hline
        GUP & 7 & 48 & \textbf{0} & \textbf{0} & \textbf{6} \\ \hline
        BICE & \textbf{5} & 16 & \textbf{0} & \textbf{0} & 10 \\ \hline
        BSX & 774 & 5155 & 171 & 666 & 1042 \\ \hline
        \plainmethod & \textbf{5} & \textbf{0} & \textbf{0} & \textbf{0} & 10 \\
    \hline
\end{tabular}
\end{table}
\end{small}

The number of unsolved queries is an important metric for evaluating the performance of subgraph matching algorithms, as it reflects the algorithm’s capability to handle difficult queries.
\tabref{tab:unsolved} reports the number of unsolved queries across five datasets, excluding HPRD, Yeast, and YouTube, where all methods, except BSX, successfully return results within the time limit due to the simplicity of these datasets (BSX encounters 0, 179, and 1236 timeouts on HPRD, Yeast, and YouTube, respectively).
A query is considered unsolved if it cannot be answered within 6 minutes. On most datasets, \method yields fewer unsolved queries than other methods, demonstrating its effectiveness in reducing redundant computation and in applying pruning strategies that prevent excessive exploration of unpromising branches.

\subsubsection{Enumeration Time Comparison Varying the Limit Number}
\label{sec:limit_number_enumeration_time_study}

To test the embedding generation speed of different algorithms, we vary the limit number from $10^5$ to $5\times 10^7$ and record the enumeration times to generate 100K, 1M, 5M, 10M, and 50M embeddings for each algorithm. 
\ifbool{fullversion}{
Note that we report the total enumeration time instead of the EPS (embeddings per second) metric in our experiments, as the mean EPS is often dominated by extreme values (detailed EPS results are provided in Appendix \ref{sec:eps_exp}).}{
Note that we report the total enumeration time instead of the EPS (embeddings per second) metric in our experiments, as the mean EPS is often dominated by extreme values (detailed EPS results are provided in the appendix of the full paper \cite{codes}).
}
\figref{fig:enum_time_comparison_diff_ret} shows the enumeration time comparison as the limit number varies. We observe that as the limit number increases, \method exhibits a lower runtime growth rate, mainly due to the black-white encoding technique, which allows \method to generate a batch of results at the same time. 
In contrast, GuP lacks a grouping technique, while DAF and RM only group matches of leaf vertices, causing their enumeration time to increase significantly with the limit number.

\subsubsection{Memory Usage Comparison}

Memory consumption is a critical factor in the practical deployment of subgraph matching algorithms. In this experiment, we compare the average peak memory consumption of \method with that of other methods. The results are presented in Table \ref{tab:memory_usage}. To reduce the overhead of memory allocation during execution, we pre-allocate a fixed amount of memory. Consequently, \method consumes more memory than other methods for small data graphs (e.g., Yeast, Human, and HPRD). However, as the data graph size increases, the memory usage of \method becomes comparable to DAF and RM, and is slightly lower in practice due to a more efficient implementation. Compared to VEQ and GuP, \method uses less memory because it does not require storing additional information in the auxiliary structure $\mathcal{A}$ for pruning.

\begin{small}
\begin{table}[ht]
    \caption{Average Peak Memory Usage (MB)}
    \label{tab:memory_usage}
  \resizebox{\linewidth}{!}{
  \begin{tabular}{ | c | c | c | c | c | c | c | c | } 
    \hline
    Datasets & CEMR & DAF & RM & VEQ & GuP & BICE & BSX \\ \hline
    Yeast & 30.8 & 6.0 & 6.4 & 9.6 & \textbf{5.2} & 9.9 & 33.2 \\ \hline
    Human & 36.3 & \textbf{7.7} & 9.6 & 15.2 & 12.2 & 15.4 & 18.2 \\ \hline
    HPRD & 42.1 & 28.0 & 9.2 & 16.6 & \textbf{8.3} & 18.3 & 10.0 \\ \hline
    WordNet & 116.2 & \textbf{27.4} & 62.8 & 125.8 & 130.1 & 151.8 & 1249.5 \\ \hline
    DBLP & 92.2 & \textbf{84.5} & 110.9 & 261.1 & 112.4 & 248.1 & 178.0 \\ \hline
    EU2005 & \textbf{536.6} & 639.8 & 726.3 & 1194.6 & 881.8 & 1399.6 & 1287.0 \\ \hline
    YouTube & \textbf{297.5} & 372.3 & 349.1 & 838.2 & 351.7 & 771.6 & 488.3 \\ \hline
    Patents & \textbf{1028.0} & 1189.6 & 1336.7 & 3180.8 & 1553.8 & 3169.9 & 2082.4 \\ \hline
\end{tabular}
}
\end{table}
\end{small}

\subsection{Ablation Studies of \plainmethod}
\label{sec:ablation_studies}

In this section, we conduct experiments to evaluate the effectiveness of our proposed techniques, including \texttt{CEM}, \texttt{CER}, pruning techniques, and matching orders.

\subsubsection{Effectiveness of \texttt{CEM}}

\begin{figure}[htbp]
    \centering
    \begin{subfigure}[t]{0.48\linewidth}
	    \includegraphics[width=\linewidth]{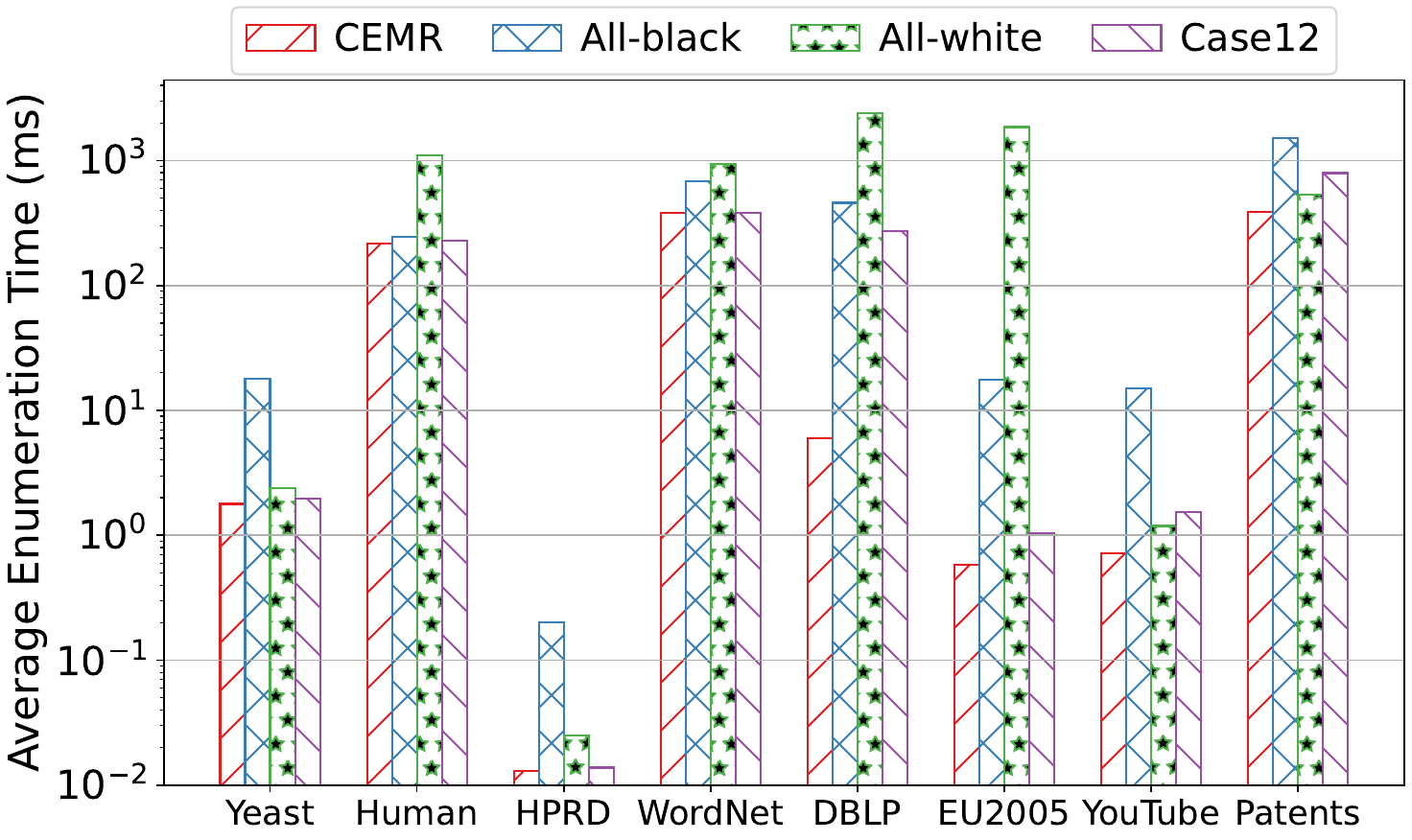}
        \caption{Ablation study of \texttt{CEM} under different vertex encodings.}
		\label{fig:ablation_encoding}
	\end{subfigure}
    \hspace{0.1cm}
    \begin{subfigure}[t]{0.48\linewidth}
        \includegraphics[width=\linewidth]{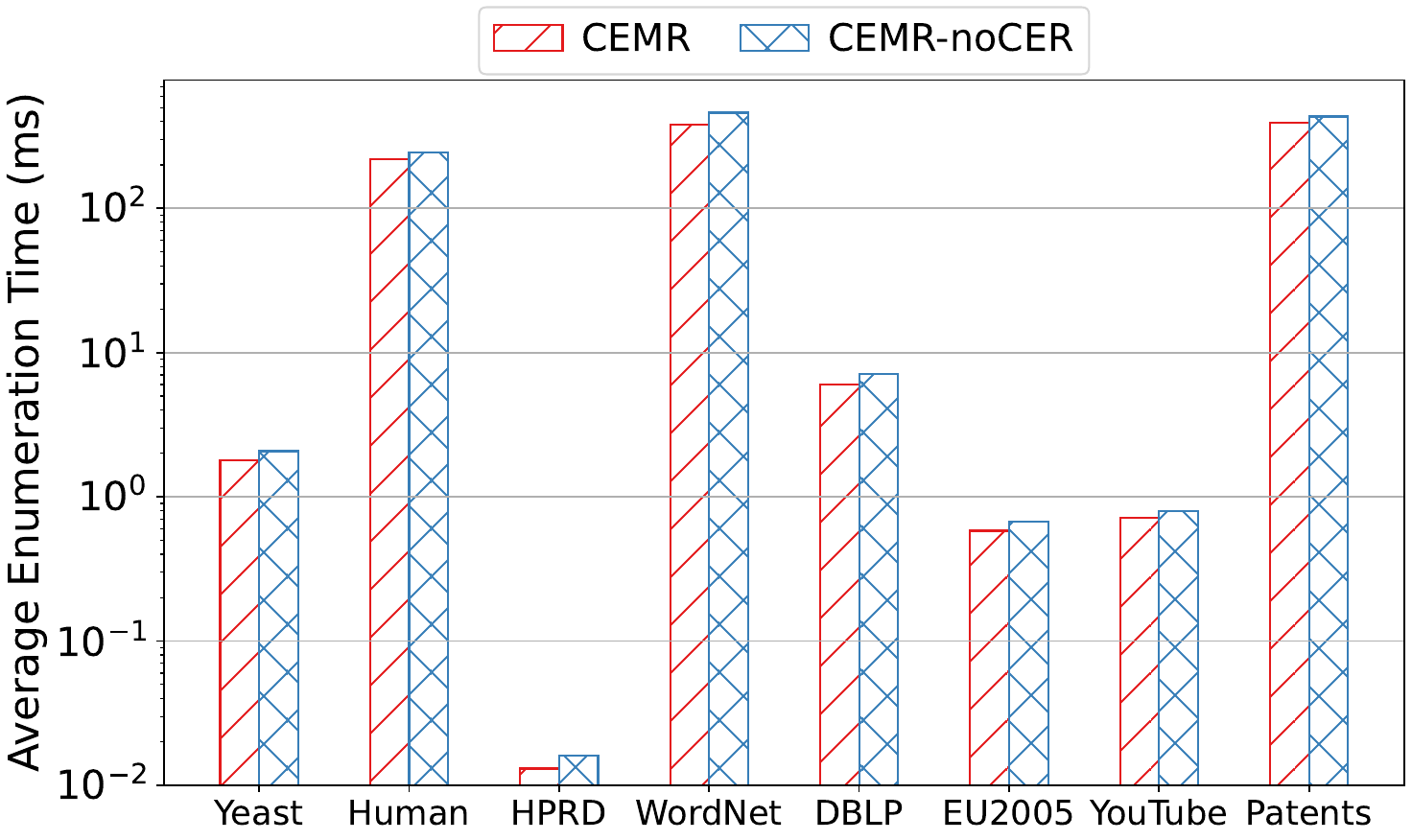}
        \caption{Ablation study of \texttt{CER}. 
        }
		\label{fig:ablation_cer}
    \end{subfigure}

    \begin{subfigure}[t]{0.48\linewidth}
	    \includegraphics[width=\linewidth]{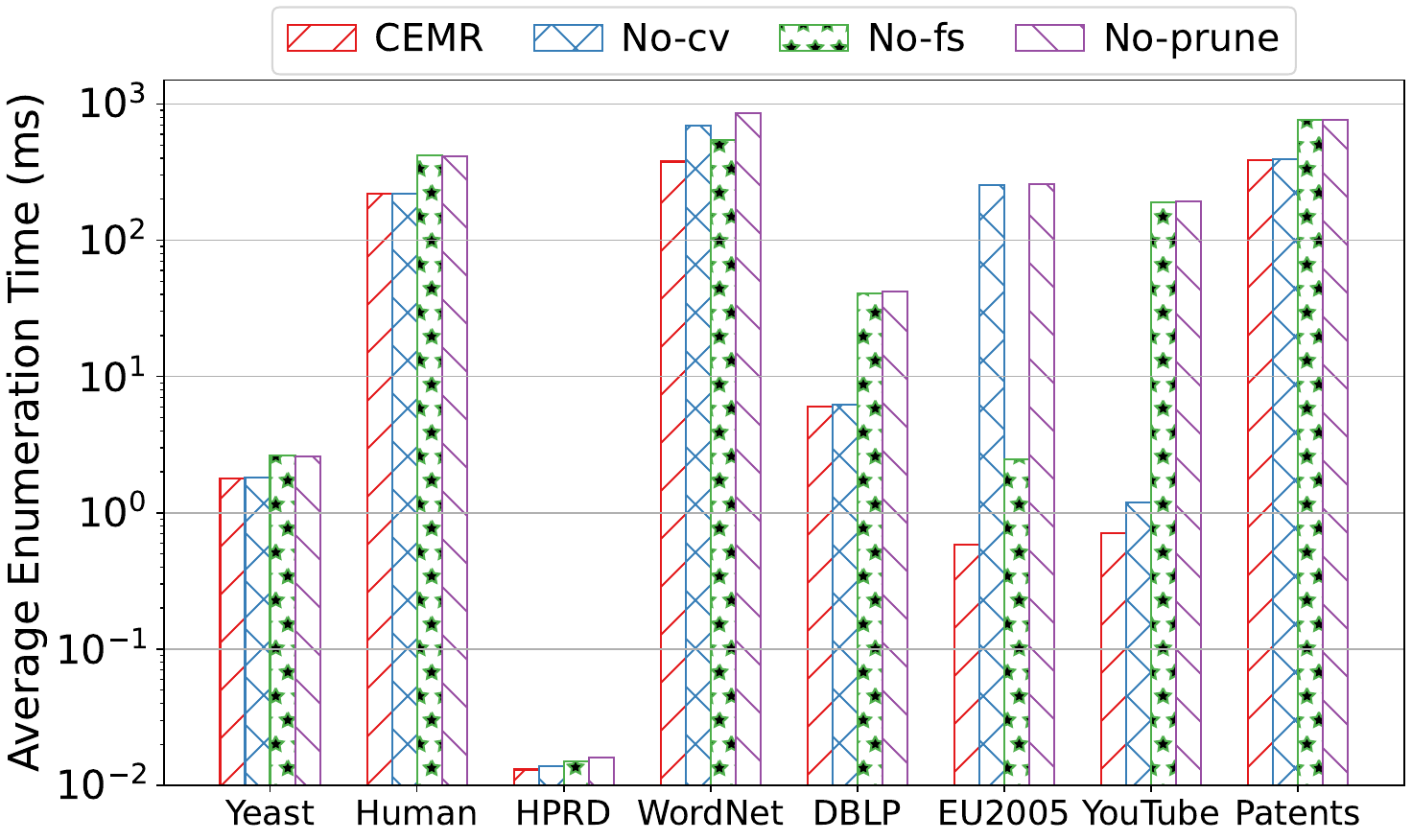}
        \caption{Ablation study of pruning techniques.}
		\label{fig:ablation_pruning}
	\end{subfigure}
    \hspace{0.1cm}
    \begin{subfigure}[t]{0.48\linewidth}
        \includegraphics[width=\linewidth]{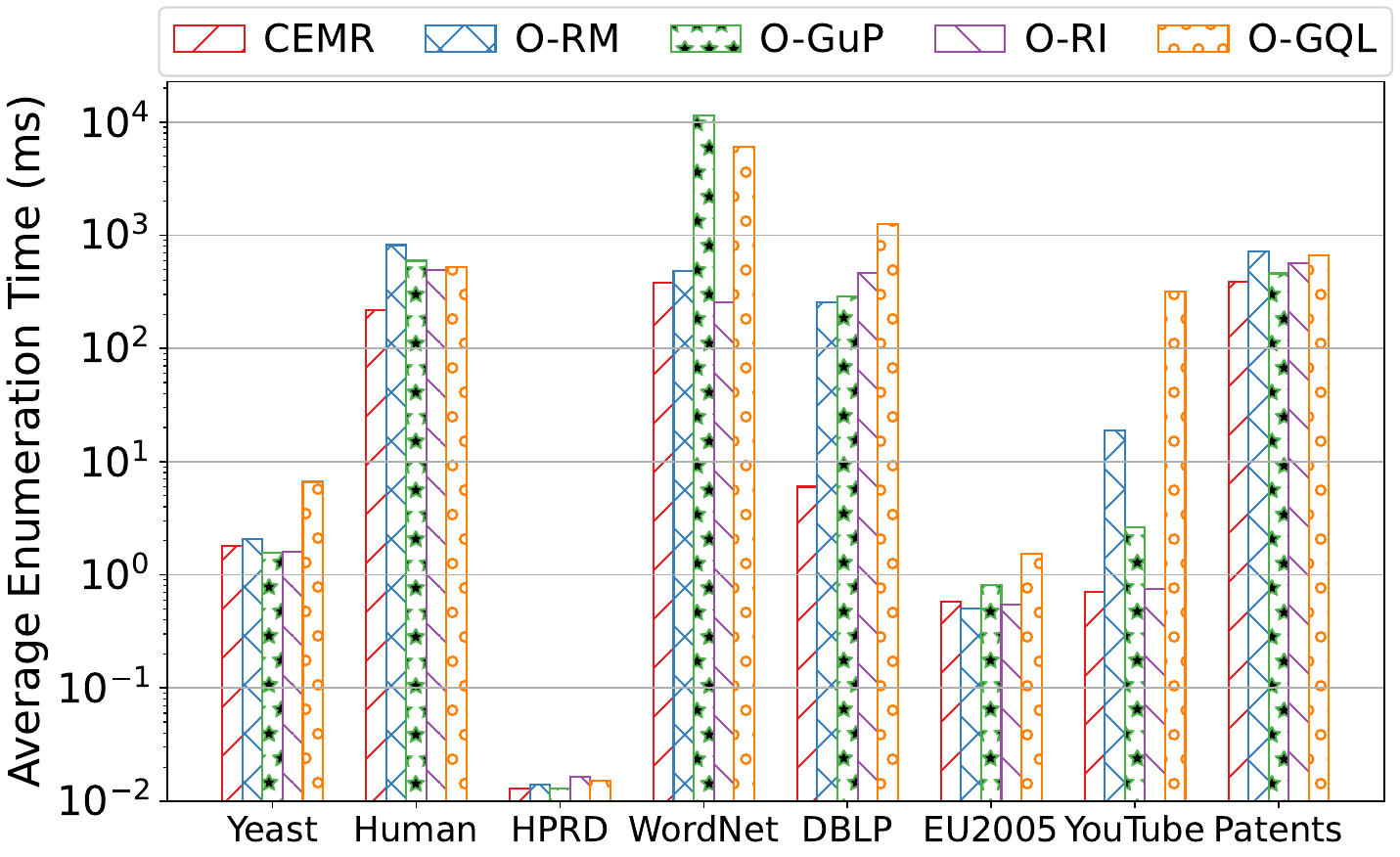}
        \caption{Ablation study of matching order. \texttt{O-*} uses orders from other methods.}
		\label{fig:ablation_order}
    \end{subfigure}
	\caption{Ablation studies.}
	\label{fig:ablation_study}
    \vspace{-0.2cm}
\end{figure}

We design a black-white vertex encoding scheme for common extension merging (\texttt{CEM} in \secref{sec:merge}), and propose an encoding strategy guided by a cost model as described in \secref{sec:encoding_strategy}.
To investigate the effect of \texttt{CEM}, we compare the average enumeration time of \method with three variants, as shown in \figref{fig:ablation_encoding}. 
In \figref{fig:ablation_encoding}, \texttt{All-black} and \texttt{All-white} encode every query vertex as black and white, respectively, while \texttt{Case12} encodes only the query vertices without forward neighbors as white, so there are no case 3 or case 4 extensions in \texttt{Case12}.

On easy datasets (e.g., Yeast, HPRD, YouTube, and Patents), most queries generate a large number of embeddings, and only a few search branches fail to yield valid results. 
In such scenarios, the aggregated mappings of white vertices help produce results more quickly, so both \method and \texttt{All-white} accelerate the enumeration process compared with \texttt{All-black}. 
On difficult datasets (e.g., Human, WordNet, DBLP, and EU2005), many redundant candidates arise from complex query topologies or vertex conflicts. 
Here, mapping certain query vertices to a single data vertex can better expose opportunities for pruning, and thus \texttt{All-black} performs better than \texttt{All-white}. 
We also observe that \texttt{Case12} consistently outperforms \texttt{All-black}. 
Overall, since \method selects vertex encoding based on a cost model, it achieves the best performance among all four strategies.

\subsubsection{Effectiveness of \texttt{CER}}


\texttt{CER} uses common extension buffers to reuse extension results across different search branches. \texttt{CEMR-noCER} in \figref{fig:ablation_cer} disables the \texttt{CER} component of \method. The results in \figref{fig:ablation_cer} show that \texttt{CER} accelerates the enumeration process, achieving a 1.11x-1.23x speedup. This improvement comes from the reduced redundant computations and the efficient sharing of partial embeddings among search branches.

\subsubsection{Effectiveness of Two Pruning Techniques}
We conduct an ablation study on \method with three variants: \texttt{No-cv}, \texttt{No-fs}, and \texttt{No-prune}, which denote the variants without contained vertex pruning (\secref{sec:contained_vertex_pruning}), extended failing set pruning (\secref{sec:failing_set_pruning}), and both pruning techniques, respectively. As shown in \figref{fig:ablation_pruning}, the full version of \method consistently outperforms all variants, indicating that both pruning techniques effectively reduce the search space and improve overall efficiency.

\subsubsection{Effectiveness of Matching Order}
To evaluate the matching order described in \secref{sec:matching_ordering}, we replace the original order in \method with those generated by RM \cite{sun2020rapidmatch} and GuP \cite{arai2023gup}. In addition, we include two well-established ordering methods recommended by a widely used survey \cite{sun2020memory}, namely RI \cite{bonnici2013subgraph} and GQL \cite{he2008graphs}. We do not compare with DAF, VEQ, BICE, or BSX, since they adopt dynamic ordering strategies, where the next matched vertex may vary during execution.
From the results shown in \figref{fig:ablation_order}, we observe that the matching order produced by \method leads to more stable performance in most cases. Moreover, RM and GuP usually achieve faster enumeration time than those in \figref{fig:total_time_comparison} (up to 203.9x). This indicates that these methods can also benefit from the redundant computation reduction techniques proposed in \method.

\subsection{Experiments on LSQB}

To better demonstrate the effectiveness of our method, we conduct evaluations beyond comparisons with state-of-the-art subgraph query algorithms. Specifically, we compare \method with the high-performance open-source graph database K{\`u}zu \cite{feng2023kuzu} on the LSQB benchmark \cite{mhedhbi2021lsqb}. LSQB is a subgraph query benchmark derived from LDBC SNB \cite{erling2015ldbc}, which consists of directed graphs with both vertex and edge labels and focuses on complex multi-join queries typical in social network analysis. It contains 9 queries, and we modified the optional and negative clauses in q7, q8, and q9 into standard clauses, resulting in 9 connected query graphs. We did not set a running time threshold and terminated execution only when the method retrieved all matches.

As shown in \figref{fig:lsqb_exp}, which reports the total enumeration time aggregated over all nine LSQB queries, \method consistently outperforms K{\`u}zu across all data scales, achieving speedups of 2.12x to 4.00x. These results demonstrate \method's efficacy on directed, multi-labeled complex workloads. The performance advantage stems from several factors: (1) our preprocessing-enumeration strategy uses effective filtering to prune unpromising search branches, reducing computational costs; and (2) as a lightweight prototype, \method avoids the overhead inherent in full-featured graph databases like K{\`u}zu.

\begin{figure}
    \setlength{\abovecaptionskip}{0.2cm}
    \centering
    \includegraphics[width=0.7\linewidth]{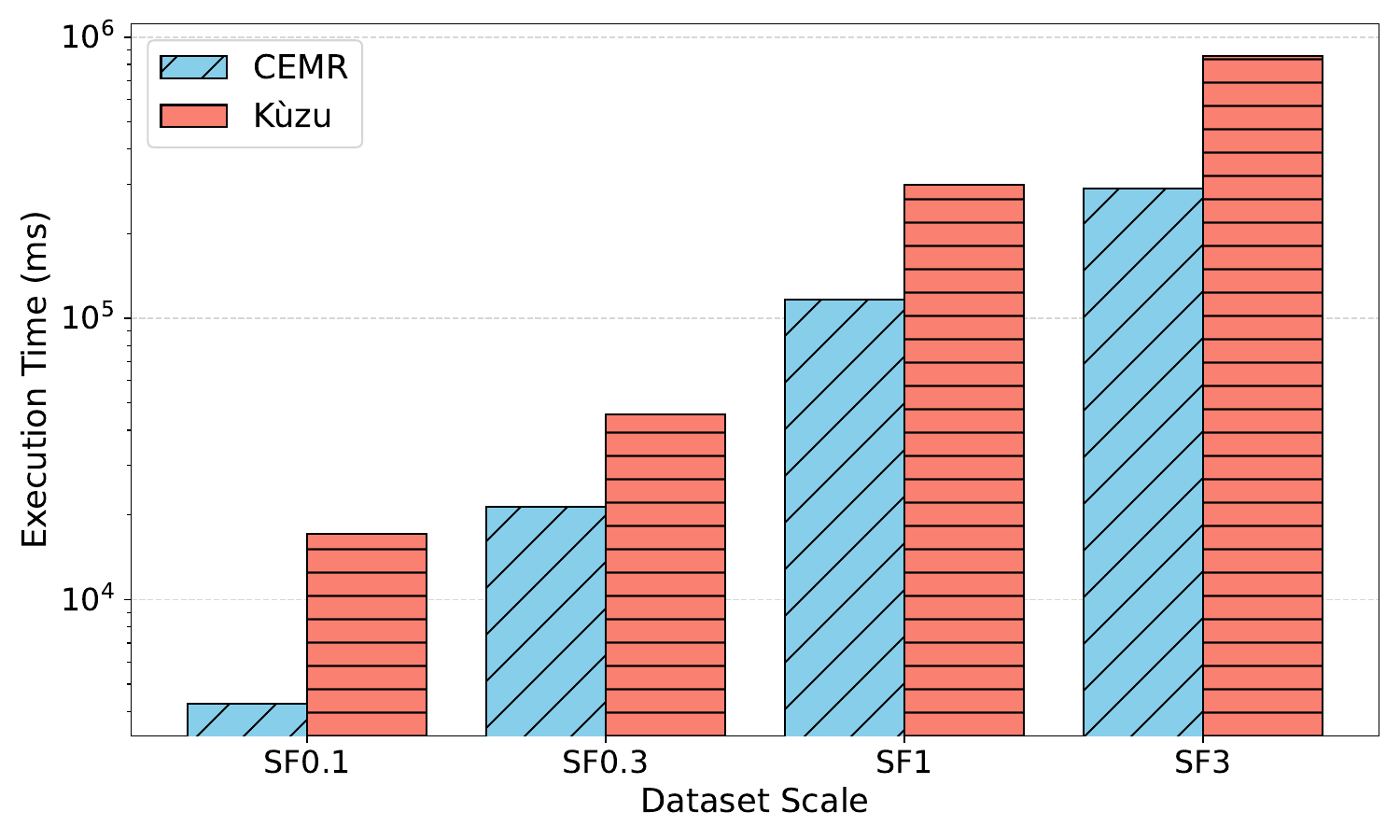}
    \caption{Execution time comparison on LSQB across different data scales.}
    \label{fig:lsqb_exp}
    \vspace{-0.2cm}
\end{figure}

\section{Conclusion}
\label{sec:conclusion}

In this paper, we propose a subgraph matching algorithm \method for redundant extension elimination. We propose a forward-looking optimization \texttt{CEM} based on black-white vertex encoding and a backward-looking optimization \texttt{CER} based on common extension buffers. They effectively improve the performance of enumeration. The experimental results show the effectiveness of \method.

\begin{acks}
This work was supported by New Generation Artificial Intelligence National Science and Technology Major Project under grant 2025ZD0123304 and NSFC under grant 62532001. This is a research achievement of National Engineering Research Center of New Electronic Publishing Technologies. The corresponding author of this work is Lei Zou (zoulei@pku.edu.cn).
\end{acks}

\clearpage

\bibliographystyle{ACM-Reference-Format}
\bibliography{sample}

\ifbool{fullversion}{
    \appendix 
    
\section{Pseudocode of Enumeration Combining \texttt{CEM} and \texttt{CER}}
\label{sec:pse_combined_algo}

\algref{algo:hybrid_enumerate} in \secref{sec:merge} only incorporates the \texttt{CEM} technique, while the \texttt{CER} technique is described in textual form in \secref{sec:reuse}. In this appendix, we present the complete pseudocode of the enumeration method that integrates both \texttt{CEM} and \texttt{CER}, as shown in \algref{algo:combined_algo}. For clarity, several components of the algorithm are extracted into helper functions, including \texttt{CompExtensions}, \texttt{CacheBuf}, and \texttt{ReuseBuf}.

\begin{small}
\begin{algorithm}
    \caption{Enumerate($Q$, $\mathcal{A}$, $O$, $\mathcal{M}$, $M$, $i$) (black-white enumeration framework with common extension reusing)}
    \label{algo:combined_algo}
    \SetAlgoVlined
    \SetAlgoVlined
    \DontPrintSemicolon
    \KwIn{The query $Q$, auxiliary structure $\mathcal{A}$, matching order $O$, result set $\mathcal{M}$, an (aggregated) embedding $M$ of $Q_{i}$, and the backtracking depth $i$.}
    \SetKwProg{myproc}{Procedure}{}{}

    \SetKwFunction{BackTrack}{BackTrack}
    \SetKwFunction{ReuseBuf}{ReuseBuf}
    \SetKwFunction{CacheBuf}{CacheBuf}
    \SetKwFunction{CompExtensions}{CompExtensions}
    \SetKwFunction{Enumerate}{Enumerate}
    \If{$i$ = $|V(Q)|$}{
        Append valid full embeddings in $M$ to $\mathcal{M}$; \textbf{return}\;
    }
    \lIf{$u_i.f \land CEB(u_i).g$} {
        $\mathcal{E}\leftarrow$\ReuseBuf($O$, $M$, $i$)
    }
    \Else{
        $\mathcal{E}\leftarrow$\CompExtensions($Q, \mathcal{A}, O, M ,i$)\;
        \lIf{$u_i.f$}{
            \CacheBuf($\mathcal{E}$, $O$, $i$)
        }
    }
    \ForEach{$M' \in \mathcal{E}$}{
        \Enumerate($Q$, $\mathcal{A}$, $O$, $\mathcal{M}$, $M'$, $i+1$)\;
        \lForEach{$u_j \in u_i.child$}{$CEB(u_j).g\leftarrow \texttt{false}$}
    }

    \BlankLine
    \myproc{\CompExtensions{$Q, \mathcal{A}, O, M ,i$}}{
    \If(\tcp*[f]{Case 1 or Case 2}){$WT(u_i) = \emptyset$}{
        $R_M(u_i)\leftarrow \cap_{u_j\in BK(u_i)}\mathcal{A}_{u_i}^{u_j}(M[u_j])$\;
        \lIf{$c(u_i) = black$}{\Return $\{M \oplus v \mid v\in R_M(u_i)\}$}
        \lElse{\Return $\{M \oplus R_M(u_i)\}$}
    }
    \Else(\tcp*[f]{Case 3 or Case 4}){
        \If{$BK(u_i) \neq \emptyset$} {
            $R_M(u_i)\leftarrow \cap_{u_j\in BK(u_i)}\mathcal{A}_{u_i}^{u_j}(M[u_j])$\;
        }
        \Else {
            $u_j\leftarrow$ $\arg\min_{u\in WT(u_i)} |M[u]|$\;
            $R_M(u_i)\leftarrow\cup_{v \in M[u_j]} \mathcal{A}_{u_i}^{u_j}(v)$\;
        }
        $\mathcal{E} \leftarrow \emptyset$\;
        \If(\tcp*[f]{Case 3}){$c(u_i)=black$} {
            \ForEach{$v\in R_M(u_i)$} {
                $M_v\leftarrow M$\;
                \ForEach{$u_j \in WT(u_i)$}{
                    $M_v[u_j] \leftarrow M[u_j] \cap \mathcal{A}_{u_j}^{u_i}(v)$\;
                }
                \If{$\forall u_j \in WT(u_i), M_v[u_j] \neq \emptyset$}{
                    $\mathcal{E}\leftarrow \mathcal{E}\cup \{M_v\oplus v\}$\;
                }
                
            }
        }
        \Else(\tcp*[f]{Case 4}){
            $\mathcal{S} \leftarrow$decompose $WT(u_i)$'s mappings\;
            \If(\tcp*[f]{Case 4.1}){$|\mathcal{S}|  \geq|R_M(u_i)|$}{
                Same pseudo-code as lines 23-28\;
            }
            \Else(\tcp*[f]{Case 4.2}){
                \ForEach{$M_t \in \mathcal{S}$}{
                    $R_{M_t}(u_i)\leftarrow \cap_{u_j\in N_-^O(u_i)}\mathcal{A}_{u_i}^{u_j}(M_t[u_j])$\; 
                    $\mathcal{E}\leftarrow \mathcal{E}\cup \{M_t\oplus R_{M_t}(u_i)\}$\;
                }
            }
        }
        \Return $\mathcal{E}$\;
    }  
    }
    
    \BlankLine
    \myproc{\CacheBuf{$\mathcal{E}, O, i$}}{
        Clear $CEB(u_i).b$\;
        \If(\tcp*[f]{Case 1 or Case 2}){$WT(u_i) = \emptyset$}{
            $R_M(u_i)\leftarrow \cup_{M'\in \mathcal{E}}M'[u_i]$\;
            Cache $R_M(u_i)$ in $CEB(u_i).b$\;
        }
        \Else(\tcp*[f]{Case 3 or Case 4}){
            Cache $\mathcal{E}$ in $CEB(u_i).b$\;
        }
        $CEB(u_i).g \leftarrow \texttt{true}$
    }

    \BlankLine
    \myproc{\ReuseBuf{$O, M ,i$}}{
        \If(\tcp*[f]{Case 1 or Case 2}){$WT(u_i) = \emptyset$}{
            \lIf{$c(u_i)=black$}{\Return $\{M \oplus v \mid v\in CEB(u_i).b\}$}
            \lElse{\Return $\{M \oplus CEB(u_i).b\}$}
        }
        \Else(\tcp*[f]{Case 3 or Case 4}){
            Recover $\mathcal{E}$ from $CEB(u_i).b$\;
            \Return $\mathcal{E}$\;
        }
    }
\end{algorithm}
\end{small}

Specifically, \texttt{CompExtensions} computes the set of possible extensions $\mathcal{E}$ for a partial embedding $M$, following the same procedure as in \algref{algo:hybrid_enumerate}. If $u_i.f = \texttt{true}$, the computed extensions are cached for future reuse by invoking \texttt{CacheBuf} (line 6). When the algorithm later encounters a brother embedding branch, the cached extensions can be directly retrieved via \texttt{ReuseBuf} (line 3). The boolean flag $CEB(u_i).g$, which indicates whether the buffer is valid, is set to \texttt{true} when \texttt{CacheBuf} is invoked (line 45), and is reset to \texttt{false} when the mapping of the parent of $u_i$ changes (line 9).

\section{Theoretical Analysis}

\subsection{Reduced Search Space Size by \texttt{CEM}}

\texttt{CEM} merges multiple search branches, thereby reducing the number of invocations of the \texttt{Enumerate} function (shown in \algref{algo:combined_algo}) as well as the computation of extensible sets. In this section, we provide a theoretical analysis of the search space reduction achieved by \texttt{CEM}. We first introduce several definitions.

\begin{definition}[Search tree and search node]
Given a matching order $O$ of a query graph $Q$, the \emph{search tree} for subgraph matching of $Q$ on a data graph $G$ is a tree structure in which each node represents a partial embedding. For clarity, we use the term \emph{node} to denote a vertex in the search tree corresponding to a partial embedding. 
For a node $M$, we use $T_M$ to denote the search subtree rooted at $M$, and denote its size (excluding $M$ itself) by $S_M$.
\end{definition}

\begin{definition}[Extension width]
\label{def:extension width}
Given a query graph $Q$, a matching order $O = (u_0, u_1, \dots, u_{|V(Q)|-1})$, and a partial embedding $M$, suppose the next query vertex to be extended is $u_i$.
For any query vertex $u_k$ with $k \ge i$ along this extension process, the \emph{extension width} from $M$ to $u_k$, denoted as $W_M(u_k)$, is defined as the number of search tree nodes at the level corresponding to $u_k$ in the subtree $T_M$.
Among these nodes, the number of distinct mappings of $u_k$ at this level is denoted as $D_M(u_k)$.
\end{definition}

\begin{example}
\begin{figure}
    \centering
    \begin{subfigure}[c]{0.34\columnwidth}
        \centering
        \includegraphics[width=0.8\linewidth]{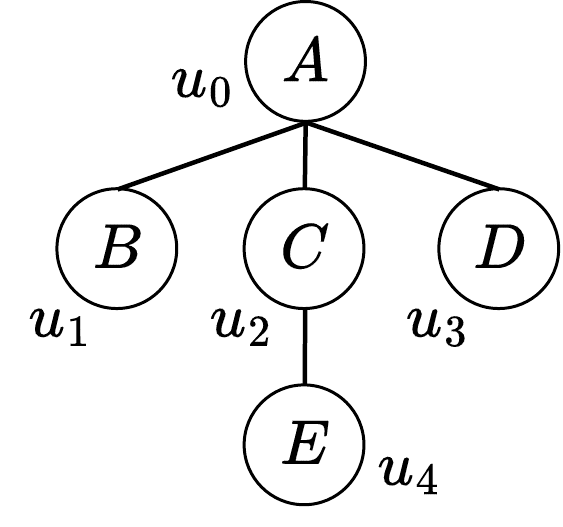}
        \caption{Query graph.}
        \label{fig:query_graph_new}
    \end{subfigure}
    \hspace{0.2cm}
    \begin{subfigure}[c]{0.49\columnwidth}
        \centering
        \includegraphics[width=0.8\linewidth]{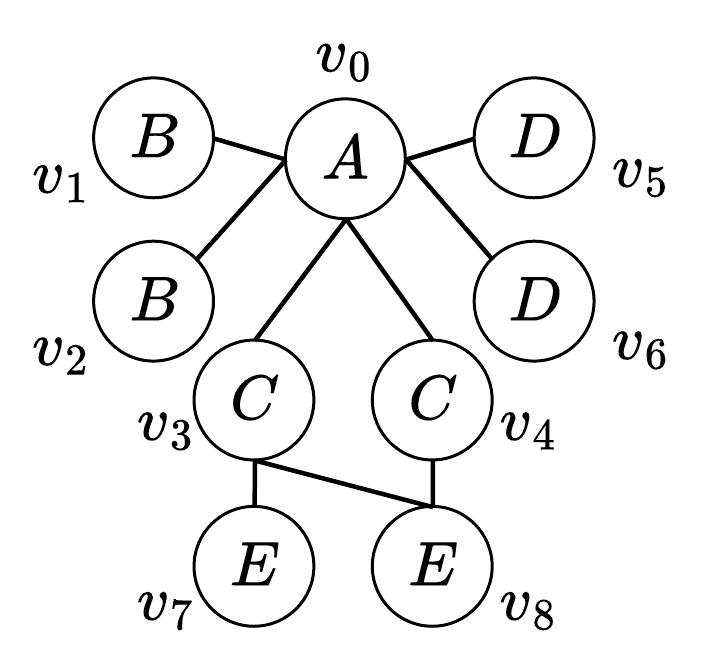}
        \caption{Data graph.}
        \label{fig:data_graph_new}
    \end{subfigure}
    \begin{subfigure}[c]{\columnwidth}
        \centering
        \includegraphics[width=0.9\linewidth]{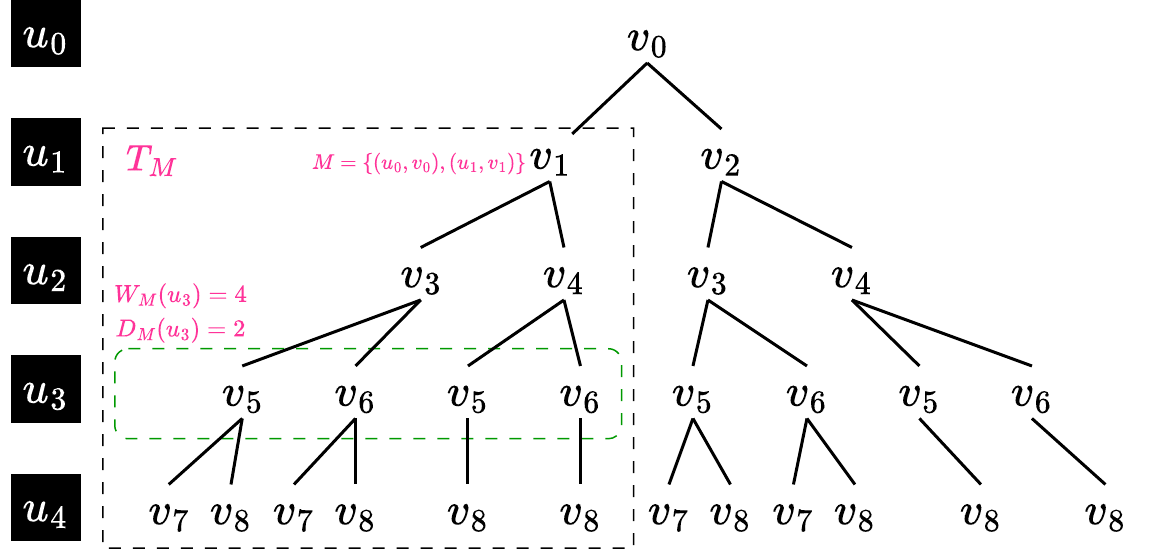}
        \caption{An example of search tree (All vertices are black).}
        \label{fig:search_tree_all_black}
    \end{subfigure}
    \caption{Illustration of search tree.}
    \label{fig:appendix_example}
\end{figure}

\figref{fig:appendix_example} illustrates a search tree rooted at $\{(u_0, v_0)\}$\footnote{There is actually a virtual node $\emptyset$ above $\{(u_0, v_0)\}$, which is used to organize the search space into a single search tree when the first query vertex has multiple candidates.}. 
For simplicity, we omit the complete partial embedding and only display the newly added mapping at each level.
At the level corresponding to query vertex $u_1$, the node labeled by $v_1$ represents the partial embedding $M = \{(u_0, v_0), (u_1, v_1)\}$.
For this embedding, we have $W_M(u_3) = 4$ and $D_M(u_3) = 2$, since there are four nodes at the $u_3$ level in $T_M$, while the distinct mappings of $u_3$ in $T_M$ are only $v_5$ and $v_6$.
\end{example}

Using the above definitions, we now analyze the reduction in the number of search tree nodes achieved by \texttt{CEM} compared with the case without \texttt{CEM}, which corresponds to the traditional method where all query vertices are treated as black.

\begin{theorem}
\label{the:search_size}
Suppose we are extending a partial embedding $M$, and the next vertex to be extended is a white vertex $u_i$\footnote{If $u_i$ is black, there is no difference from the traditional method shown in \algref{algo:set_intersection_framework}.}. Then the following results hold.
\begin{itemize}[leftmargin=10pt]
    \item If $u_i$ has no forward neighbor, let the size of the search subtree $T_M$ under the traditional method be denoted by $S_M$. When employing \texttt{CEM}, the size of the search subtree is reduced to $S_M / |R_M(u_i)|$.
    \item If $u_i$ has a forward neighbor $u_k$, the vertices $\{u_j \mid i < j < k\}$ do not connect to $u_i$, and $u_k$ is either a black vertex (Case 3) or a white vertex extended in Case 4.1, then under the traditional method the extension width at level $u_j$ is $W_M(u_j)$ for $i \le j \le k$. When employing \texttt{CEM}, the extension width is reduced to $W_M(u_j) / |R_M(u_i)|$ for $i \le j < k$, and becomes $(W_M(u_{k-1}) / |R_M(u_i)|) \times D_M(u_k)$ at level $u_k$.
    \item If $u_i$ has a forward neighbor $u_k$, the vertices $\{u_j \mid i < j < k\}$ do not connect to $u_i$, and $u_k$ is a white vertex extended in Case 4.2, then the extension width is reduced to $W_M(u_j) / |R_M(u_i)|$ for $i \le j < k$, and becomes $(W_M(u_{k-1}) / |R_M(u_i)|) \times \prod_{u_l \in WT(u_k)} |M[u_l]|$ at level $u_k$.
\end{itemize}
\end{theorem}

\begin{example}
\begin{figure}
    \begin{subfigure}[c]{0.48\columnwidth}
        \centering
        \includegraphics[width=0.95\linewidth]{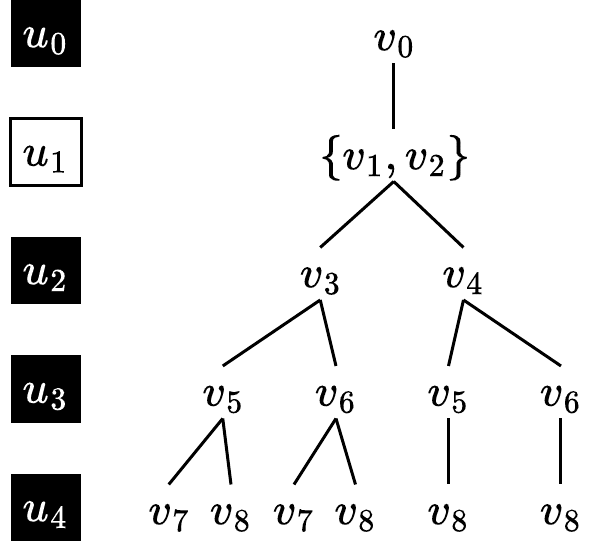}
        \caption{Search tree ($u_1$ is white).}
        \label{fig:appendix_search_tree_2}
    \end{subfigure}
    \hspace{0.1cm}
    \begin{subfigure}[c]{0.48\columnwidth}
        \centering
        \includegraphics[width=0.62\linewidth]{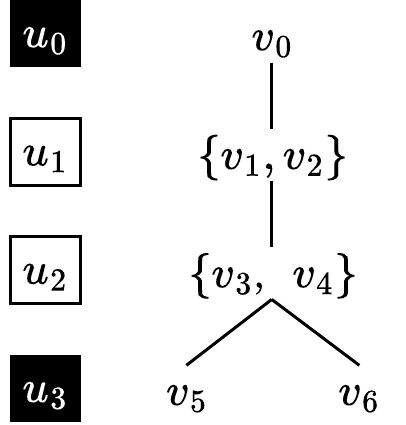}
        \caption{Partial search tree ($u_1$ and $u_2$ are white).}
        \label{fig:appendix_search_tree_3}
    \end{subfigure}
    \vspace{0.1cm}
    \begin{subfigure}[c]{0.48\columnwidth}
        \centering
        \includegraphics[width=\linewidth]{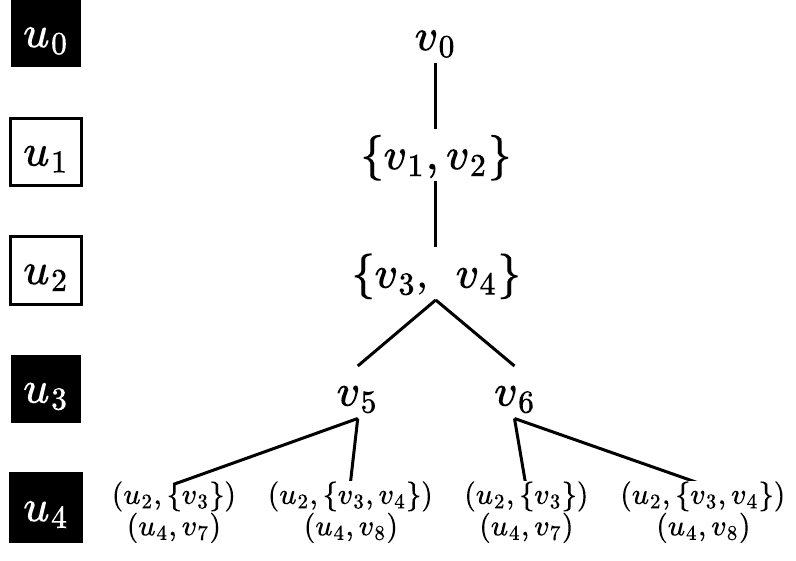}
        \caption{Search tree ($u_1$ and $u_2$ are white).}
        \label{fig:appendix_search_tree_4}
    \end{subfigure}
    \hspace{0.1cm}
    \begin{subfigure}[c]{0.48\columnwidth}
        \centering
        \includegraphics[width=\linewidth]{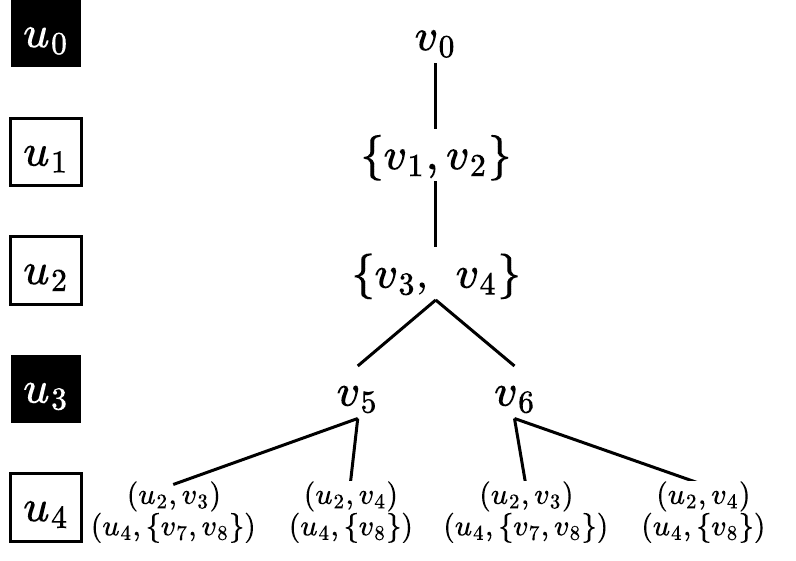}
        \caption{Search tree ($u_1$, $u_2$, and $u_4$ are white).}
        \label{fig:appendix_search_tree_5}
    \end{subfigure}
    \caption{Examples illustrating the reduction of search trees achieved by \texttt{CEM} under different vertex encoding.}
    \label{fig:appendix_search_tree_again_example}
    \vspace{-0.4cm}
\end{figure}
\figref{fig:appendix_search_tree_again_example} illustrates \thmref{the:search_size} using three small examples.
\begin{itemize}[leftmargin=10pt]
    \item Compared with the search tree in \figref{fig:search_tree_all_black}, merging the mappings of the white vertex $u_1$ reduces the search subtree size $S_{\{(u_0, v_0)\}}$ by half, since $|R_{\{(u_0, v_0)\}}(u_1)| = 2$ (\figref{fig:appendix_search_tree_2}).
    \item When both $u_1$ and $u_2$ are encoded as white, the last level of the search tree in \figref{fig:appendix_search_tree_4} is derived from the partial search tree in \figref{fig:appendix_search_tree_3}. Since the mapping of $u_4$ is used to trim the mappings of the white vertex $u_2$, we use abbreviated notation in the figure for simplicity. Let $M = \{(u_0, v_0), (u_1, \{v_1, v_2\})\}$. Compared with \figref{fig:appendix_search_tree_2}, encoding $u_2$ as white reduces the extension widths $W_M(u_2)$ and $W_M(u_3)$ by half, and the extension width at level $u_4$ becomes $(4/2) \times 2 = 4$, instead of 6 in \figref{fig:appendix_search_tree_2}.
    \item If $u_4$ is further encoded as white, as shown in \figref{fig:appendix_search_tree_5}, the extension width at the final level is also reduced to $(4/2) \times 2 = 4$.
\end{itemize}
\end{example}

\subsection{Space Complexity of \texttt{CER}}

We analyze the space complexity of extending a partial embedding $M$ with a query vertex $u$, assuming that $u.f = \texttt{true}$, which indicates that the common extension buffer for $u$ is enabled.
\begin{itemize}[leftmargin=10pt]
    \item If $u$ falls into Case 1 or Case 2, we store $R_M(u)$ in $CEB(u).b$. The resulting space complexity is $O(|R_M(u)|)$, which is bounded by $O(|C(u)|)$ or $O(|V(G)|)$.
    \item If $u$ falls into Case 3, we need to store $\mathcal{E} = \{ M_v \oplus v \mid v \in R_M(u) \}$ (see \appref{sec:pse_combined_algo}), where $M_v$ denotes the reduced embedding derived from $M$ and associated with $v \in R_M(u)$. The space complexity is therefore $O(\sum_{v \in R_M(u)} |M_v \oplus v|)$, which is bounded by $O(|C(u)| \times |V(G)|^{|V(Q)|-1})$, and hence $O(|V(G)|^{|V(Q)|})$ in the worst case, since the match size of each white vertex in $M_v$ is at most $O(|V(G)|)$.
    \item If $u$ falls into Case 4, the space complexity is identical to that of Case 3, namely $O(|V(G)|^{|V(Q)|})$.
\end{itemize}
Note that the above analysis represents the worst case. In practice, the actual space consumption is typically much smaller due to effective filtering and pruning strategies.

\section{Proofs of Lemmas}

\subsection{Proof of \lemref{lemma:redundant}}
\label{proof:redundant}

\begin{proof}
If $M_1 \oplus v$ is a valid match of $Q_{i+1}$, then by the definition of extensible vertices, we have $v \in \bigcap_{u \in N_-^O(u_i)} \mathcal{A}_{u_i}^u(M_1[u])$.
Since $M_1[u] = M_2[u]$ for each $u \in N_-^O(u_i)$, $v$ is also extensible for $u_i$ under $M_2$. Given that $v$ does not conflict with any vertex in $M_2$, we conclude that $M_2 \oplus v$ is also a valid match of $Q_{i+1}$.
\end{proof}

\subsection{Proof of \lemref{lem:contained_vertex_pruning}}
\label{proof:pruning}
\begin{proof}
We first claim that the extensible vertices for each $u_j \in Con(u_i)$ must be a subset of $R_M(u_i)$. This follows from the following containment chain:
\begin{align}
    \bigcap_{u\in N_-^O(u_j)} \mathcal{A}_{u_j}^u(M'[u]) 
    &\subseteq \bigcap_{u\in N_-^O(u_j)\cap \{u_0, \cdots, u_{i-1}\}} \mathcal{A}_{u_j}^u(M[u]) \notag \\
    &\subseteq \bigcap_{u\in N_-^O(u_i)} \mathcal{A}_{u_i}^u(M[u]) \subseteq R_M(u_i),
\end{align}
where $M'$ denotes an extension of $M$ (i.e., $M \subseteq M'$). This means that any extensible vertex of $u_j$ must belong to $R_M(u_i)$. 

If $|R_M(u_i)| < |Con(u_i)|$, one of the following two situations must occur during the future enumeration:
\begin{enumerate}[leftmargin=14pt]
    \item Some $u_j \in Con(u_i)$ will have an empty extensible set;
    \item Two vertices in $Con(u_i)$ are forced to share the same mapping in $R_M(u_i)$, violating the injectivity constraint.
\end{enumerate}
In either case, no valid embedding can be constructed, and thus the current search branch can be safely pruned.
\end{proof}

\section{Experiments Using EPS Metric}
\label{sec:eps_exp}

The embeddings-per-second (EPS) metric is not included in the main paper due to its sensitivity to extreme values. For completeness, we report the EPS comparison in \figref{fig:eps_comp}. Overall, \method achieves higher EPS than other methods in most cases, indicating competitive execution efficiency.

For BSX, the EPS values can be strongly influenced by a few extreme cases. For instance, on the EU2005 dataset, certain queries result in more than $10^{53}$ embeddings in a single batch, which can lead to misleading averages. While BSX experiences a number of timeout cases, their effect on the averaged EPS is negligible in comparison to that of extreme values.  

\begin{figure}
    \centering
    \begin{subfigure}[c]{0.85\linewidth}
        \centering
        \includegraphics[width=\linewidth]{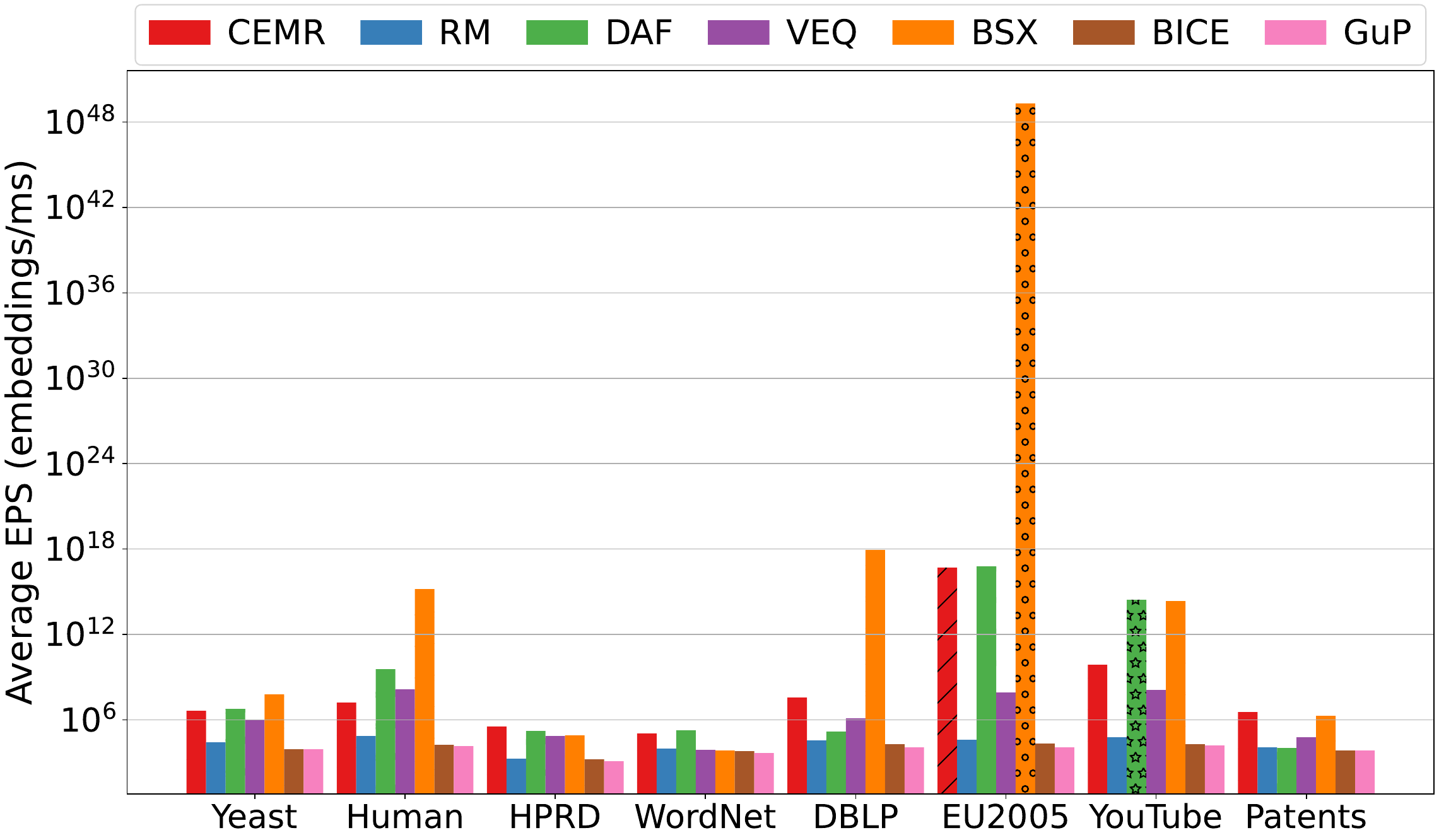}
        \captionsetup{skip=1.0pt}
        \caption{Mean EPS shown as a bar plot.}
        \label{fig:eps_bar}
    \end{subfigure}
    \begin{subfigure}[c]{0.85\linewidth}
        \centering
        \includegraphics[width=\linewidth]{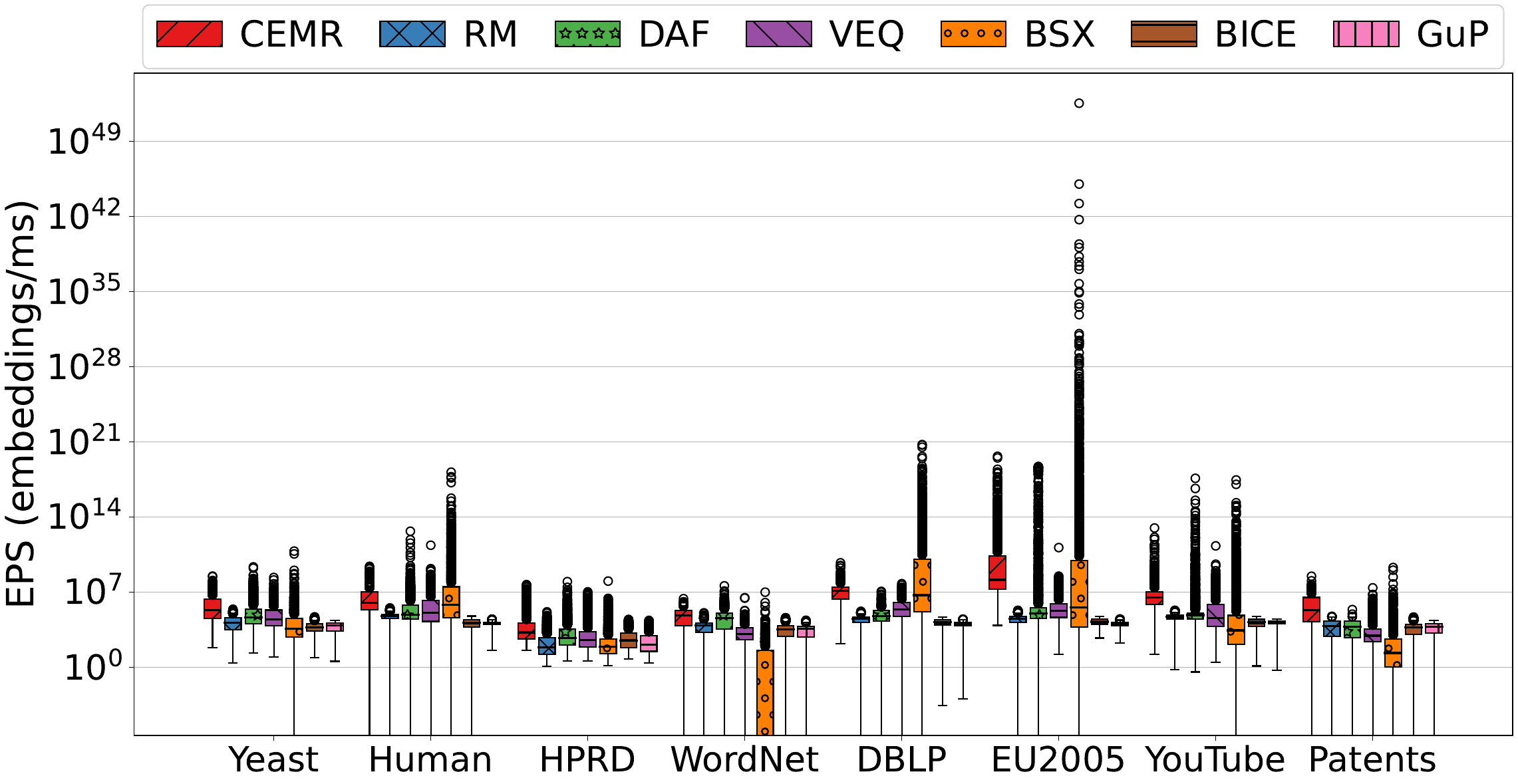}
        \captionsetup{skip=1.0pt}
        \caption{Distribution of EPS shown as a box plot, with outliers marked as circles.}
        \label{fig:eps_box}
    \end{subfigure}
    \captionsetup{skip=1.1pt}
    \caption{EPS comparison. EPS is set to zero for timeout cases.}
    \label{fig:eps_comp}
\end{figure}

\section{Discussions}
\label{sec:discussion_extend_to_system}

Modern graph database systems commonly rely on worst-case-optimal joins (WCOJ) \cite{ngo2018worst} as the core execution paradigm for subgraph queries. WCOJ evaluates multiway joins in a vertex-at-a-time manner and has been shown to achieve optimal worst-case complexity for join processing, making it particularly suitable for subgraph matching workloads.

It is straightforward to incorporate \method into existing engines that employ DFS-based WCOJ, such as gStore for processing limit-k queries \cite{yang2022gcbo}, where limit-k queries refer to queries containing a \texttt{LIMIT} clause. In these systems, query execution incrementally expands a partial match by introducing one additional query vertex at each step. This execution pattern is essentially identical to enumeration-based subgraph query processing, as illustrated in \algref{algo:generic_framework} \cite{sun2020rapidmatch}.

Within this class of systems, \method can be integrated by extending the standard \textit{ExtendOneNode} operator. When the operator processes a query vertex $u$, we augment it with two auxiliary indicators. The first indicator specifies the extending case of $u$, corresponding to the four cases defined in \secref{sec:four_cases}. The second indicator determines whether the \texttt{CER} mechanism (\secref{sec:reuse}) is enabled for $u$ during extension.

Although \method is designed for DFS-based enumeration, it can also be applied to systems that originally adopt BFS-based enumeration by switching their execution mode to DFS-based. Moreover, some systems employ hybrid execution plans that combine worst-case-optimal joins with binary joins, such as K{\`u}zu \cite{feng2023kuzu} and EmptyHeaded \cite{aberger2017emptyheaded}. In these systems, \method can be applied to pure WCOJ subplan trees, where no binary joins are involved, or where the binary join appears only as the left-deep operator.

}

\end{document}